\documentclass[10pt,conference]{IEEEtran}
\usepackage{graphicx}
\usepackage{amsthm}
\usepackage{epsfig}
\usepackage{latexsym}
\usepackage{amsfonts}
\usepackage{here}
\usepackage{rawfonts}
\usepackage[T1]{fontenc}
\usepackage{calc}
\usepackage{capitalgreekitalic}
\usepackage{url}
\usepackage{enumerate}
\usepackage{color}
\usepackage[tbtags]{amsmath}
\usepackage{amssymb}
\usepackage{upref}
\usepackage{epic,eepic}
\usepackage{times}
\usepackage{dsfont}
\usepackage{comment}
\usepackage{cite}
\usepackage{mathrsfs}
\usepackage{verbatim}
\usepackage{float}
\usepackage{chngcntr}
\usepackage{bbm}
\usepackage{hyperref}
\allowdisplaybreaks
\usepackage{caption}
\usepackage{subcaption}
\allowdisplaybreaks
\usepackage{algorithm}
\usepackage{algpseudocode}
\DeclareMathOperator*{\argmax}{argmax}

\newtheorem{theorem}{\bf Theorem}
\newtheorem{proposition}{\bf Proposition}
\let\oldproposition\proposition
\renewcommand{\proposition}{\oldproposition\normalfont}

\newtheorem{definition}{\bf Definition}
\let\olddefinition\definition
\renewcommand{\definition}{\olddefinition\normalfont}

\newlength{\totlinewidth}
\newcounter{substep}

	{\end{list}}

\newlength{\aligntop}
\newlength{\alignbot}
 \makeatletter

\IEEEoverridecommandlockouts

\begin{document}
\title{\huge  Inter-Operator Resource Management for Millimeter Wave, Multi-Hop Backhaul Networks}
\author{
\authorblockN{Omid Semiari$^\dag$, Walid Saad$^\dag$, Mehdi Bennis$^\ddag$, and Zaher Dawy$^*$\\}
\authorblockA{\small $^\dag$Wireless@VT, Bradley Department of Electrical and Computer Engineering, Virginia Tech, Blacksburg, VA, USA,\\
	  Email: \protect\url{{osemiari, walids}@vt.edu}\\
\small $^\ddag$Centre for Wireless Communications-CWC, University of Oulu, Finland, Email: \url{bennis@ee.oulu.fi}\\
\small $^*$  Department of Electrical and Computer Engineering, American University of Beirut, Lebanon, Email: ~\url{zd03@aub.edu.lb}
    \thanks{This research was supported by the U.S. National Science Foundation under Grants CNS-1460316 and CNS-1617896.}%
  }
\maketitle
\begin{abstract}
In this paper, a novel framework is proposed for optimizing the operation and performance of a large-scale, multi-hop millimeter wave (mmW) backhaul within a wireless small cell network (SCN) that encompasses multiple mobile network operators (MNOs). The proposed framework enables the small base stations (SBSs) to jointly decide on forming the multi-hop, mmW links over backhaul infrastructure that belongs to multiple, independent MNOs, while properly allocating resources across those links.  In this regard, the problem is addressed using a novel framework based on matching theory that is composed to two, highly inter-related stages: a multi-hop network formation stage and a resource management stage. One unique feature of this framework is that it jointly accounts for both wireless channel characteristics and economic factors during both network formation and resource management. The multi-hop network formation stage is formulated as a one-to-many matching game which is solved using a novel algorithm, that builds on the so-called deferred acceptance algorithm and is shown to yield a stable and Pareto optimal multi-hop mmW backhaul network. Then, a one-to-many matching game is formulated to enable proper resource allocation across the formed multi-hop network. This game is then shown to exhibit peer effects and, as such, a novel algorithm is developed to find a stable and optimal resource management solution that can properly cope with these peer effects. Simulation results show that, with manageable complexity, the proposed framework yields substantial gains, in terms of the average sum rate, reaching up to $27\%$ and $54\%$, respectively, compared to a non-cooperative scheme in which inter-operator sharing is not allowed and a random allocation approach. The results also show that our framework improves the statistics of the backhaul sum rate and provides insights on how to  manage pricing and the cost of the cooperative mmW backhaul network for the MNOs.
\end{abstract}

\section{Introduction}\label{intro}
Network densification based on the concept of small cell networks (SCNs) is seen as the most promising solution to cope with the increasing demand for wireless capacity \cite{17}. SCNs are built on the premise of a viral and dense deployment of small base stations (SBSs) over large geographical areas so as to reduce the coverage holes and improve the spectral efficiency\cite{Quek13}. However, such a large-scale deployment of SBSs faces many challenges in terms of resource management, network modeling, and backhaul support  \cite{Quek13}.

In particular, providing backhaul support for a large number of SBSs that can be deployed at adverse locations within a geographical area has emerged as one of the key challenges facing the effective operation of future heterogeneous SCNs  \cite{interdigital}. In particular, due to the density of SCNs, mobile network operators (MNOs) will not be able to maintain an expensive and costly deployment of fiber backhauls to service SBSs as shown in \cite{Hur13} and \cite{interdigital}. Instead, MNOs are moving towards the adoption of wireless backhaul solutions that are viewed as an economically viable approach to perform backhauling in dense SCNs. In fact, MNOs expect that $80 \%$ of SBSs will connect to the core network via wireless backhaul as detailed in \cite{interdigital} and \cite{6189405}.

Existing works have proposed a number of solutions for addressing a handful of challenges facing SCN backhauling \cite{Rangan14,Boccardi01,Hur13,Ghosh14,interdigital,6189405,Liu15,7110547,7248860,Liebl01,Yi01,Loumiotis01,7324401,7414178,7324420,7305743,7331861,7194083,Sengupta01,Mahloo01,7389369,interdigital14}. The authors in \cite{Liebl01} propose a fair resource allocation for the out-band relay backhaul links. The proposed approach developed in \cite{Liebl01} aims to maximize the throughput fairness among backhaul and access links in LTE-Advanced relay system. In \cite{Yi01}, a backhaul resource allocation approach is proposed for LTE-Advanced in-band relaying. This approach optimizes resource partitioning between relays and macro users, taking into account both backhaul and access links quality. \textcolor{black}{Dynamic backhaul resource provisioning is another important problem in order to avoid outage in peak traffic hours and under-utilizing frequency resources in low traffic scenarios. In this regard, in \cite{Loumiotis01}, a dynamic backhaul resource allocation approach is developed based on evolutionary game theory. Instead of static backhaul resource allocation, the authors take into account the dynamics of users' traffic demand and allocate sufficient resources to the base stations, accordingly.} Although interesting, the body of work in \cite{Liebl01,Yi01,Loumiotis01} does not consider the potential deployment of millimeter wave communication at the backhaul network and is primarily focused on modeling rather than resource management and multi-hop backhaul communication.

Providing wireless backhaul links for SBSs over \textit{millimeter wave} (mmW) frequencies has recently been dubbed as one of the most attractive technologies for sustaining the backhaul traffic of SCNs \cite{Rangan14,Boccardi01,Hur13,Ghosh14,interdigital,6189405,Liu15,7110547,7248860}, \textcolor{black}{due to the following promising characteristics, among others:
	 1) The mmW spectral band that lies within the range 30-300 GHz will deliver high-capacity backhaul links by leveraging up to $10$ GHz of available bandwidth which is significantly larger than any ultra-wideband system over sub-6 GHz frequency band. In addition, high beamforming gains are expected from mmW antenna arrays, with large number of elements, to overcome path loss \cite{Liu15}, 
	 2) more importantly, mmW backhaul links will not interfere with legacy sub-6 GHz communications in either backhaul or access links, due to operating at a different frequency band. Even if the access network operates over the mmW frequency  band such as in self-backhauling architectures, mmW communications will generally remain less prone to interference, due to the directional transmissions, short-range links, as well as susceptibility to the blockage \cite{6840343},
	  and 3) over the past few years,  research for utilizing mmW frequencies for wireless backhaul networks has become an interesting field that attracted a lot of attention in both academia and industry \cite{Rangan14,Boccardi01,Hur13,Ghosh14,interdigital,6189405,Liu15,7110547,7248860,MiWaveS}. As an example, in 2014, a total of 15 telecom operators, vendors, research centers, and academic institutions (including Nokia, Intel, and operators Orange and Telecom Italia) have launched a collaborative  project in Europe, called \emph{MiWaveS}, to develop mmW communications for 5G backhaul and access networks \cite{MiWaveS}.}
  
\textcolor{black}{However, compared to existing ultra-dense networks over sub-6 GHz band, the major \emph{challenges} of mmW backhaul networks can be listed as follows:
	1) MmW backhaul links will typically operate over much shorter range than their sub-6 Ghz counterparts (usually do not exceed 300 meters \cite{7414178,rappaport2014}), and, thus, more SBSs will be required to provide backhaul support for the users within a certain geographical area. Therefore, mmW SBS deployments are expected to be even denser, compared to the already dense sub-6 GHz networks \cite{7422408}. Such ultra dense network will require fast and efficient network formation algorithms to establish a multi-hop backhaul link between the core network and each demanding SBS, 2) the backhaul network must be significantly reliable. However, the received signal power of mmW signals may significantly degrade if the backhaul link is blocked by an obstacle. For SBSs that are deployed in adverse locations, such as urban furniture, the received signal power may degrade due to rain or blockage by large vehicles. One solution is to increase the density of SBSs such that if a backhaul link between two SBSs is blocked, the demanding SBS can establish a reliable link with another SBS. However, this solution will increase the cost of the backhaul network for the MNO. In our work, we have motivated the use of cooperation between MNOs to achieve a robust and economically efficient backhaul solution, and 	
	3) due to the directional transmissions of the mmW signals, broadcast control channels can lead to a mismatch between the control and data planes at mmW frequency bands \cite{Niu111}. Therefore, fully centralized approaches that rely on receiving control signals from a central station over broadcast channels may not be practical, thus, motivating the adoption of suitable distributed algorithms for an effective resource management.}

Several recent studies have studied the viability of mmW as a backhaul solution as presented in \cite{7110547} and \cite{7194083,7324401,7414178,7324420,7331861,7305743}. For instance, the work in \cite{7110547} proposes a model based on stochastic geometry to analyze the performance of the self-backhauled mmW networks. The work in \cite{7324401} analyzes the performance of a dual-hop backhaul network for mmW small cells. In \cite{7414178}, the authors perform channel measurements and provide insights for the mmW small cell backhaul links. In \cite{7324420}, the performance of  adaptive and switching beamforming techniques
are investigated and evaluated for mmW backhaul networks. Moreover, the impact of diffraction loss in mmW backhaul network is analyzed in \cite{7331861}. The authors in \cite{7305743} propose a multi-objective optimization framework for joint deployment of small cell base stations and wireless backhaul links. In \cite{7194083}, the authors propose an autonomous beam alignment technique for self-organizing multi-hop mmW backhaul networks.
\textcolor{black}{In \cite{7422408}, the authors have motivated the use of a multi-hop mmW backhaul as a viable solution for emerging $5$G networks and they analyzed  the impact of the deployment density on the backhaul network capacity and power efficiency. Moreover, in \cite{narlikar}, the authors have proposed a multi-hop backhaul solution with a TDMA MAC protocol for WiMAX.}

The body of work in \cite{7110547} and \cite{7194083,7324401,7414178,7324420,7331861,7305743} solely focuses on physical layer metrics, such as links' capacity and coverage. In addition, it is focused only on single-hop or dual-hop backhaul networks, while new standards such as IEEE 802.11ay envision fully multi-hop networks. \textcolor{black}{The work presented in \cite{7422408} does not provide any algorithm to determine how SBSs must form a multi-hop mmW backhaul network. Moreover, the proposed model in \cite{7422408} is too generic and does not capture specific characteristics of a mmW network, such as susceptibility to blockage and directional transmissions. Last but not the least, no specific analysis or algorithm is provided for resource management in multi-hop mmW backhaul networks. The solution presented in \cite{narlikar} is not directly applicable to the mmW backhaul networks, as mmW is substantially different from WiMAX systems. In fact, authors in \cite{narlikar} focus primarily on the routing and link activation protocols in order to minimize the interference among active links. Such a conservative approach will yield an inefficient utilization of the mmW frequency resources, since interference scenario in WiMAX systems is completely different with directional mmW communications.}

Furthermore, the body of work in \cite{7110547}, \cite{7194083,7324401,7414178,7324420,7331861,7305743}, \cite{7422408}, and \cite{narlikar} does not account for the effect of backhaul cost in modeling backhaul networks. In fact, these existing works typically assume that all infrastructure belong to the same MNO which may not be practical for dense SCNs. In wireless networks, the backhaul cost constitutes a substantial portion of the total cost of ownership (TCO) for MNOs as indicated in  \cite{Hur13} and \cite{interdigital}. In fact, it is economically inefficient for an individual MNO to afford the entire TCO of an independent backhaul network as demonstrated in \cite{Hur13}, \cite{interdigital}, and \cite{interdigital14}. 
\textcolor{black}{The main advantages of inter-operator backhaul sharing is to reduce the number of required sites/radio access technology (RAT) interfaces per MNO to manage backhaul traffic, site rent, capital expenditures (CAPEX) by avoiding duplicate infrastructure, site operating expenditures (OPEX), and electricity costs \cite{taga}.}  
Moreover, inter-operator mmW backhaul architectures are more robust against the blockage and link quality degradation compared to the schemes in which operators act independently and non-cooperatively \cite{Rangan14}. This stems from the fact that cooperation increases flexibility to establish new backhaul links that can easily bypass obstacles. Therefore, MNOs will need to share their backhaul network resources with other MNOs that require backhaul support for their SBSs \cite{interdigital14}. Hence, beyond the technical challenges of backhaul management in SCNs, one must also account for the cost of sharing backhaul resources between MNOs.  

To address such economic challenges, a number of recent works have emerged in \cite{7248860} and \cite{interdigital14,Sengupta01,Mahloo01,7389369}. The work in \cite{interdigital14} motivates a business model for an SCN where multiple MNOs share the SBSs that are deployed on the street lights of dense urban areas. In \cite{Sengupta01} an economic framework is developed to lease the frequency resources to different MNOs by using novel pricing mechanisms. In \cite{Mahloo01}, the authors propose a cost evaluation model for small cell backhaul networks. This work highlights the fact that integrating heterogenous backhaul technologies is mandatory to achieve a satisfactory performance in a backhaul network. Moreover, they show that the TCO of an SCN is much higher than conventional cellular networks. Therefore, it is more critical to consider backhauling cost in small cell backhaul network design. The authors in \cite{7389369} propose a model where MNOs buy energy from the renewable power suppliers for their mmW backhaul network and solve the problem as a Stackelberg game between MNOs and power suppliers. In \cite{7248860}, we studied the problem of resource management for the mmW-microwave backhaul networks with multiple MNOs. The approach in \cite{7248860} considers both cost and the channel state information (CSI) to allocate backhaul resources to the SBSs. The provided solutions in \cite{Sengupta01,Mahloo01,7389369,interdigital14} focus solely on the economic aspects of the backhaul network, while a suitable backhaul network model must integrate the cost constraints with the physical constraints of the wireless network. In addition, \cite{7248860} does not consider multi-hop backhaul networks. Moreover, the backhaul model studied in \cite{7248860} is restricted to the case in which only two MNOs are in the network.


The main contribution of this paper is to propose a novel framework to model and analyze resource management and pricing for facilitating inter-operator sharing of multi-hop, mmW backhaul infrastructure in dense SCNs. In particular, the proposed framework is formulated using suitable techniques from  matching theory \cite{Roth92} so as to provide a distributed solution for managing the resources over multi-hop backhaul links. In the formulated model, the SBSs of one MNO can act as \emph{anchored BSs (A-BSs)} to provide backhaul support to other, \emph{demanding BSs (D-BSs)} that may belong to other MNOs. The proposed framework is composed of two highly-interrelated matching games: a network formation game and a resource management game. The goal of the network formation game is to associate the D-BSs to A-BSs for every hop of the backhaul links. This game is shown to exhibit peer effects thus mandating a new algorithmic approach that differs from classical matching works in \cite{Roth92} and \cite{eduard11}. To solve this game, we propose a distributed algorithm that is guaranteed to converge to a two-sided stable and Pareto optimal matching between the A-BSs and the D-BSs. Once the stable and optimal network formation solution is found, we propose a second matching game for resource management that allocates the sub-channels of each A-BS to its associated D-BSs, determined by the first matching game. The proposed approach considers the cost of the backhaul jointly with the links' achievable rates to allocate the sub-channels to the D-BSs. To solve this resource management matching game with peer effects, we propose a novel distributed algorithm that yields a two-sided stable and Pareto optimal matching between the sub-channels and the D-BSs. We compare the performance of the proposed \textit{cooperative mmW multi-hop backhaul network} (mmW-MBN) and compare the results with non-cooperative mmW-MBN. Simulation results show that MNOs cooperation provides significant gains in terms of network's average \textcolor{black}{backhaul sum rate}, reaching up to $30 \%$, compared to the non-cooperative mmW-MBN. The results also show that the cooperation among MNOs will significantly improve the statistics of the backhaul rate per SBS.

The rest of this paper is organized as follows. Section ~\ref{system model} describes the system model and formulates the problem. Section \ref{sec:III} presents our distributed approach to solve the network formation problem. Section \ref{Sec:IV} provides the proposed solution to solve the resource allocation problem. Section \ref{simulations} provides the simulation results and Section \ref{conclusion} concludes the paper.

\begin{figure}
	\centering
	\centerline{\includegraphics[width=7cm]{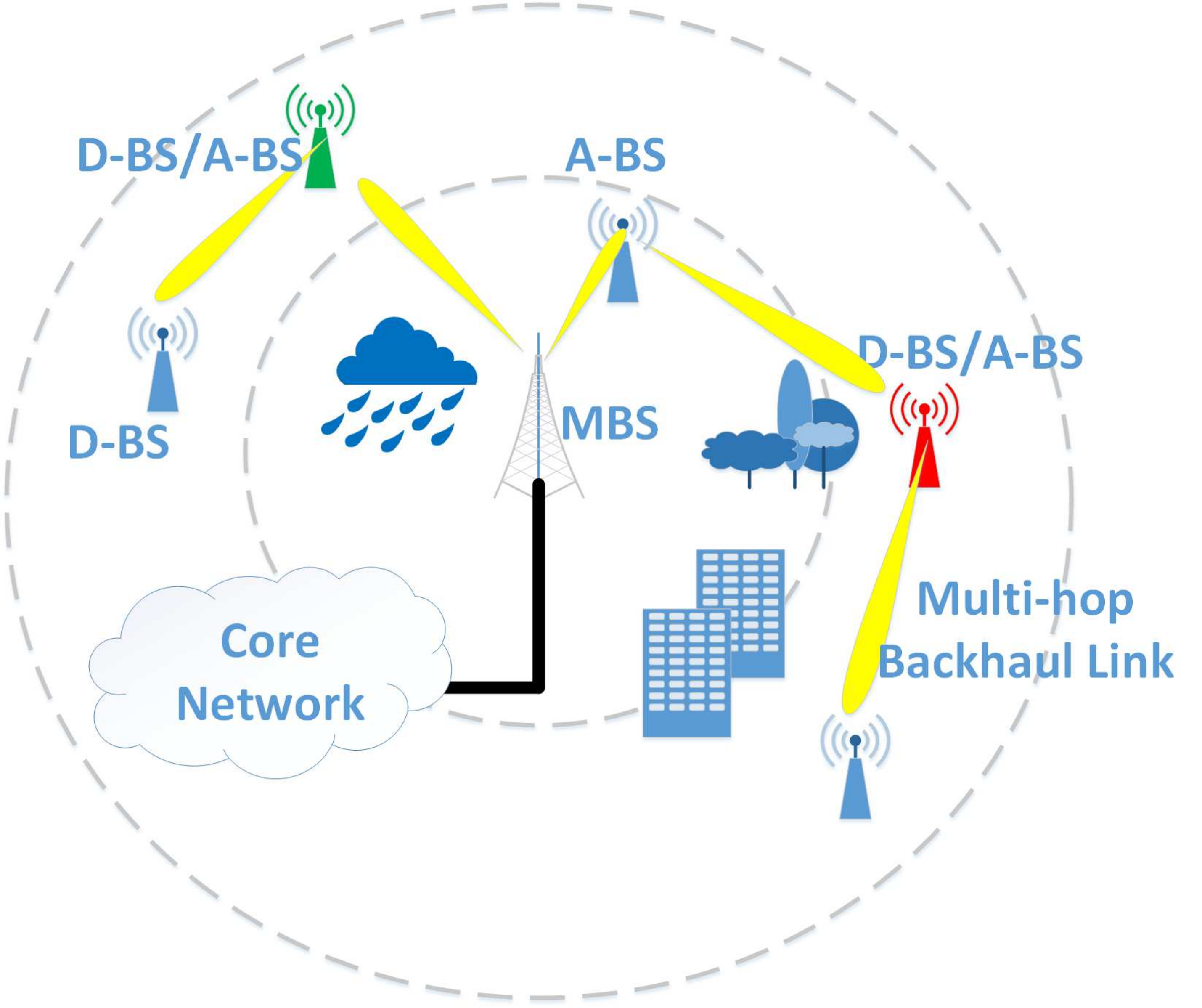}}
	\caption{\small An example of mmW-MBN with multiple MNOs. SBSs with the same color belong to the same MNO.}\vspace{0em}
	\label{model}
\end{figure}\vspace{0em}

\begin{table*}[!t]
	\footnotesize
	\centering
	\caption{Variables and notations}\vspace*{-.5em}
	\begin{tabular}{|c|c||c|c|}
		\hline
		\bf{Notation} & \bf{Description} & \bf{Notation} & \bf{Description} \\
		\hline
		$M$ & Number of SBSs & $\mathcal{M}$ & Set of SBSs\\
		\hline
		$N$ & Number of MNOs & $\mathcal{N}$ & Set of MNOs\\
		\hline
		$K$ & Number of sub-channels & $\mathcal{K}$ & Set of sub-channels\\
		\hline
		$e(m',m)$ & Backhaul link from $m'$ to $m$ & $\mathcal{E}$& Set of backhaul links\\
		\hline
		$\mathcal{M}_n$ & Set of SBSs belonging to MNO $n$ & $\mathcal{M}_m^d$ & Set of SBSs of distance $d$ from $m$\\
		\hline
		$\mathcal{M}_{m'}^{\textrm{D-BS}}$ & SBSs for whom SBS $m'$ serves as an A-BS & $w$ & Bandwidth of each sub-channel \\
		\hline
		$\zeta_e \in \{0,1\}$ & State of link $e$ & $\rho_e$ & Expected value of $\zeta_e$\\
		\hline
		$r_m(k,m';\zeta)$ & Rate for D-BS $m$ over sub-channel $k$ & $r_m(m',\boldsymbol{x})$ & Rate of D-BS $m$ from A-BS $m'$, given $\boldsymbol{x}$  \\
		\hline
		$\succ_m^D$ & Preference profile of D-BSs over A-BSs&$\succ_m^A$  & Preference profile of A-BSs over D-BSs\\
		\hline
		$P_m^D$ & Preference profile of D-BSs over sub-channels & $P_k^K$ & Preference profile of sub-channels over D-BSs\\
		\hline
		$\pi_j$ & Network formation matching for $j$-th hop & $\mu_j$ & Resource allocation matching for $j$-th hop\\
		\hline
		$Q_m$ & Quota of A-BS $m$ & $r_{m,\textrm{th}}$ & Backhaul minimum rate requirement for $m$\\
		\hline
		$\mathds{1}_{mm'}$ & Indicates if $m$ and $m'$ belong to same MNO& $\bar{r}_m(m')$ & Average rate for $m$ over all sub-channels\\
		\hline
	\end{tabular}\label{tab1}\vspace{-2em}
\end{table*}

\section{System Model}\label{system model}
Consider a mmW-MBN that is used to support the downlink transmissions of $M$ SBSs within the set $\mathcal{M}$. Each SBS belong to one of $N$ MNOs within the set $\mathcal{N}$. The set $\mathcal{M}$ can be decomposed into $N$ subsets $\mathcal{M}_n$, with $\bigcup_{n \in \mathcal{N}}\mathcal{M}_n=\mathcal{M}$ and $\bigcap_{n \in \mathcal{N}}\mathcal{M}_n=\emptyset$, where $\mathcal{M}_n$ represents the subset of SBSs belonging to MNO $n$. The SBSs are distributed uniformly in a planar area with radius $d_{\textrm{max}}$ around a macro base station (MBS), $m_0$, located at $(0,0) \in \mathbbm{R}^2$. The MBS is connected to the core network over a broadband fiber link, as shown in Fig. \ref{model}, and is shared by all MNOs. The SBSs can be connected to the MBS via a \emph{single-hop or a multi-hop mmW} link. The mmW-MBN can be represented as a directed graph $G(\mathcal{M},\mathcal{E})$, in which the SBSs are the vertices and $\mathcal{E}$ is the set of edges. Each edge, $e(m',m)\in \mathcal{E}$, represents a mmW backhaul link from SBS $m'$ to $m$. Hereinafter, for any link, the transmitting and the receiving SBSs (over the backhaul) will be referred to, respectively, as the A-BSs and the D-BSs. 

Thus, in our model, an SBS can be either a D-BS or an A-BS. Each A-BS $m$ will serve up to $Q_m$ D-BSs, while each D-BSs will be connected to one A-BS.

To show that an arbitrary D-BS $m$ is connected to an A-BS $m'$, we use the following binary variable
\begin{align}\label{mmwlink}
\epsilon_{e}(m',m) = \begin{cases}
1             &\textrm{if} \,\,\,e(m',m)\in \mathcal{E},\\
0             & \text{otherwise},
\end{cases}
\end{align}
where $\epsilon_{e}(m',m) =0$ implies that no backhaul link exists from SBS $m'$ to $m$. Finally, we denote by $\mathcal{M}^{\textrm{D-BS}}_{m'}$ the subset of SBSs for whom SBS $m'$ serves as an A-BS. In other words, $\mathcal{M}^{\textrm{D-BS}}_{m'}=\{m \in \mathcal{M}|\,\epsilon_{e}(m',m) =1\}$.
The backhaul links are carried out over a mmW frequency band, composed of $K$ sub-channels, within the set $\mathcal{K}$, each of a bandwidth $w$. A summary of our notation is provided in Table \ref{tab1}.

\subsection{Channel Model}
The state of a backhaul link is defined as a Bernoulli random variable $\zeta_{m'm}$ with success probability $\rho_{m'm}$ to determine if the link is LoS or NLoS. In fact, $\zeta_{m'm}=1$, if $e(m',m)$ is LoS, otherwise, $\zeta_{m'm}=0$. 
Based on the field measurements carried out in \cite{Ghosh14} and \cite{7070688,7510705,7501500}, the large-scale path loss of the link $e(m',m)$, denoted by $L_{\text{dB}}\left(m',m\right)$ in dB, is given by 
\begin{align}\label{pathloss}
L_{\text{dB}}(m',m)&=10\log_{10}(l(m',m))=20\log_{10}\left(\frac{4\pi d_0}{\lambda}\right)\notag\\
&+ 10\alpha \log_{10}\left(\frac{\|\boldsymbol{y}_m-\boldsymbol{y}_{m'}\|}{d_0}\right)+\chi, \,\,\, d\geq d_0,
\end{align}
where $\lambda$ is the wavelength at carrier frequency $f_c = 73$ GHz, $d_0$ is the reference distance, and $\alpha$ is the path loss exponent. Moreover, $\|\boldsymbol{y}_m-\boldsymbol{y}_{m'}\|$ is the Euclidean distance between SBSs $m$ and $m'$, located, respectively, at $\boldsymbol{y}_m \in \mathbbm{R}^2$ and $\boldsymbol{y}_{m'} \in \mathbbm{R}^2$. In addition, $\chi$ is a Gaussian random variable with zero mean and variance $\xi^2$. Path loss parameters $\alpha$ and $\xi$ will naturally have different values, depending on the state of the link. In fact, depending on whether the link is LoS or NLoS, these values can be chosen such that the path loss model in \eqref{pathloss} will provide the best linear fit with the field measurements carried out in \cite{Ghosh14}.  The benefit of the free space path loss model  used in \eqref{pathloss}, compared with other models such as the alpha-plus-beta  model, is that it is valid for all distances above the reference distance $d_0$ and the model parameters $\alpha$ and $\chi$ have concrete physical interpretations.

In addition, the field measurements in \cite{Rappaport1232,7481755,7534832} show that the mmW channel delay spread can be large, reaching up to more than $100$ ns, for the outdoor deployment of mmW SBSs in urban areas. To this end, for any link $e(m',m)$, a slow-varying frequency flat fading channel $h_{m'km}$ is considered over sub-channel $k$. Hence, conditioned to the link state $\zeta$, the achievable rate for a given link $e(m',m)$ over sub-channel $k$ will be given by
	\begin{align}\label{rate}
	&r_{m}(k,m';\zeta) = w \log_2\bigg(1+\\\notag
	&\frac{p_{m',k}\psi(m',m)l\left(m',m\right)|h_{m'km}|^2}{\sum_{m'' \neq m,m'}p_{m'',k}\psi(m'',m)l\left(m'',m\right)|h_{m''km}|^2 + \sigma^2}\bigg),
	\end{align}
	where $p_{m',k}$ and $\sigma^2$ denote, respectively, the transmit power of A-BS $m'$ over sub-channel $k$ and the noise power. To strike a balance between system performance and complexity, uniform power allocation is assumed. Here, we assume that total transmit power $p_{t,m'}$ is distributed uniformly over all sub-channels, such that $p_{m',k}=p_{t,m'}/K$ \cite{5374055,Kim2014,6692451,5733445}. The uniform power allocation assumption is also due to the fact that at a high SNR/SINR regime, as is expected in a mmW network with relatively short-range links and directional transmissions, it is well known that optimal power allocation policies such as the popular water-filling algorithm will ultimately converge to the uniform power allocation \cite{5733445}. Moreover, $\psi(m',m)$ represents the combined transmit and receive antenna gains. The antenna gain pattern for each BS is assumed to be sectorized and is given by \cite{7110547}:
	\begin{align}
	G(\theta)=\begin{cases}
	G_{\text{max}}, &\text{if} \,\,\,\,\,\theta <|\theta_m|,\\
	G_{\text{min}}, &\text{otherwise},
	\end{cases}
	\end{align}
	where $\theta$ and $\theta_m$ denote, respectively, the azimuth angle and the antennas' main lobe beamwidth. Moreover, $G_{\text{max}}$ and $G_{\text{min}}$ denote, respectively, the antenna gain of the main lobe and side lobes. It is assumed that for a desired link between A-BS $m'$ and D-BS $m$, $\psi(m',m)=G_{\text{max}}^2$. 
	Moreover, $\psi(m'',m)$ of an interference link from A-BS $m''$ to the target D-BS $m$ is assumed to be random. Using \eqref{rate}, we can write the achievable rate for the link $e(m',m)$ over the allocated sub-channels as follows:
\begin{align}\label{rate2}
r_{m}(m';\boldsymbol{x}) = \sum_{k\in \mathcal{K}}r_{m}(k,m';\zeta)x_{m'km},
\end{align}
where $\boldsymbol{x}$ is the resource allocation vector with elements $x_{m'km}=1$, if SBS $m'$ transmits to $m$ over sub-channel $k$, otherwise, $x_{m'km}=0$. In \eqref{rate2}, we remove the dependency on $\zeta$ in the left-hand side to simplify the notations. \textcolor{black}{Here, considering a decode-and-forward scheme, we note that, if an SBS $m$ is connected to the MBS via a multi-hop link of length $n$, then $r_m(m';\boldsymbol{x})$ will be limited by $1/n$ times the minimum (bottleneck) of all link rates over the multi-hop connection \cite{1638547}.} In addition, by averaging with respect to $\zeta$, the average achievable rate over all sub-channels for D-BS $m$ assigned to A-BS $m'$ will be
\begin{align}\label{rate_ave}
\bar{r}_m(m') &=\mathbbm{E}\left[\sum_{k \in \mathcal{K}}r_m(m',k;\zeta)\right],\\
&= \mathbbm{P}(\zeta_{m',m}=1)\sum_{k \in \mathcal{K}}r_m(m',k;\zeta)x_{m'km}\notag\\
&+\mathbbm{P}(\zeta_{m',m}=0)\sum_{k \in \mathcal{K}}r_m(m',k;\zeta)x_{m'km},\\
&= \rho_{m'm}\sum_{k \in \mathcal{K}}r_m(m',k;\zeta=1)x_{m'km}\notag\\
&+(1-\rho_{m'm})\sum_{k \in \mathcal{K}}r_m(m',k;\zeta=0)x_{m'km}.\label{rate_ave_7}
\end{align}

For dense urban areas, the number of obstacles blocking an arbitrary link $e(m',m)$ increases as $\|\boldsymbol{y}_m-\boldsymbol{y}_{m'}\|$ increases. Such severe shadowing will significantly reduce the received signal power, particularly, for street-level deployment of mmW SBSs over urban furniture such as lamp posts. Therefore, the communication range of each SBS will be limited to a certain distance $d$, where $d$ depends on the density of the obstacles, as suggested in \cite{7414178} and \cite{rappaport2014}. To this end, we define $\mathcal{M}_m^d$ as
\begin{align}\label{M^d}
\mathcal{M}_m^d = \{m' \in \mathcal{M}, m' \neq m\big|\,\, \|\boldsymbol{y}_m-\boldsymbol{y}_{m'}\| \leq d \},
\end{align}
which effectively represents the set of SBSs with which $m$ is able to communicate over an LoS or an NLoS link.

\subsection{Network formation and resource allocation in mmW-MBN with multiple MNOs} 
We consider a cooperative, inter-operator mmW-MBN in which, under proper pricing incentives, the SBSs of each MNO may act as A-BSs for other SBSs belonging to other MNOs. We let $q_m$ be a unit of price per sub-channel of SBS $m \in \mathcal{M}_n$, as determined by MNO $n$. That is, if $x_{mkm'}=1$ for $m \in \mathcal{M}_n$ and $m' \in \mathcal{M}_{n'}$, where $n \neq n'$, MNO $n'$ will have to pay $q_m$ to MNO $n$. To solve the resource management problem for the proposed mmW-MBN, we need to first determine the backhaul links, $\epsilon_e(m',m)$, and then specify the rate over each link, $r_m(m',\boldsymbol{x})$. To this end, as illustrated in Fig. \ref{organ}, we must solve two interrelated problems: 1) \emph{network formation} problem that determines $\mathcal{E}$,  and 2) \emph{resource allocation} problem to assign sub-channels of each A-BS $m$ to their corresponding D-BSs in $\mathcal{M}_m^{\textrm{D-BS}}$. 
 
The network formation problem can be formulated as follows: 
\begin{subequations}
\begin{IEEEeqnarray}{rCl}
\!\!\!\!\!\!\!\!\!\!\!\!\!\!\!\!\!\argmax_{\mathcal{E}} && \sum_{m \in \mathcal{M}}\sum_{m'\in \mathcal{M}_m^d}\epsilon_{e}(m',m)\bar{r}_m(m')-\kappa_m q^t_{m'}\mathds{1}_{mm'},\label{opt1:a}\\
\!\!\!\!\!\!\!\!\!\!\!\!\text{s.t.}\,\,\,\,
&&\epsilon_{e}(m',m)+\epsilon_{e}(m,m'')\notag\\
&&+\epsilon_{e}(m'',m')\leq 2, \forall m,m',m'' \in \mathcal{M},\label{opt1:b}\\
&&\!\!\!\!\!\!\sum_{\,\,\,\,\,\,\,\,\,m'\in \mathcal{M}_m^d}\!\!\!\epsilon_{e}(m',m)\leq 1, \forall m \in \mathcal{M},\label{opt1:c}\\
&&\!\!\!\!\!\!\sum_{\,\,\,\,\,\,\,\,\,m\in \mathcal{M}_{m'}^d}\!\!\!\epsilon_{e}(m',m)\leq Q_{m'},  \forall m' \in \mathcal{M},\label{opt1:d}\\
&&\epsilon_{e}(m',m)+\epsilon_{e}(m,m')\leq 1, \forall m,m' \in \mathcal{M},\label{opt1:e}\\
&&\epsilon_{e}(m',m) \in \{0,1\}, \forall m,m' \in \mathcal{M},\label{opt1:f}
\end{IEEEeqnarray}
\end{subequations}
where $\mathds{1}_{mm'}=1$, if both SBSs $m$ and $m'$ belong to different MNOs, otherwise, $\mathds{1}_{mm'}=0$. In addition, $\kappa_m$ is a weighting scalar that scales the cost of a link with respect to its rate. The total cost of a link $e(m',m)$ for $m$ is $q^t_{m'}=q_{m'}\sum_{k\in\mathcal{K}}x_{m'km}$. Constraint \eqref{opt1:b} is to avoid any cycles. In addition, \eqref{opt1:c} indicates that each D-BS must be assigned to at most one A-BS. Moreover, \eqref{opt1:d} indicates that each A-BS $m'$ can be assigned to up to $Q_{m'}$ D-BSs. Constraint \eqref{opt1:e}  ensures that all links are directional. That is, an SBS $m$ may transmit to $m'$ or receive its traffic from $m'$, however, cannot do both simultaneously. 

\begin{figure}
	\centering
	\centerline{\includegraphics[width=\columnwidth]{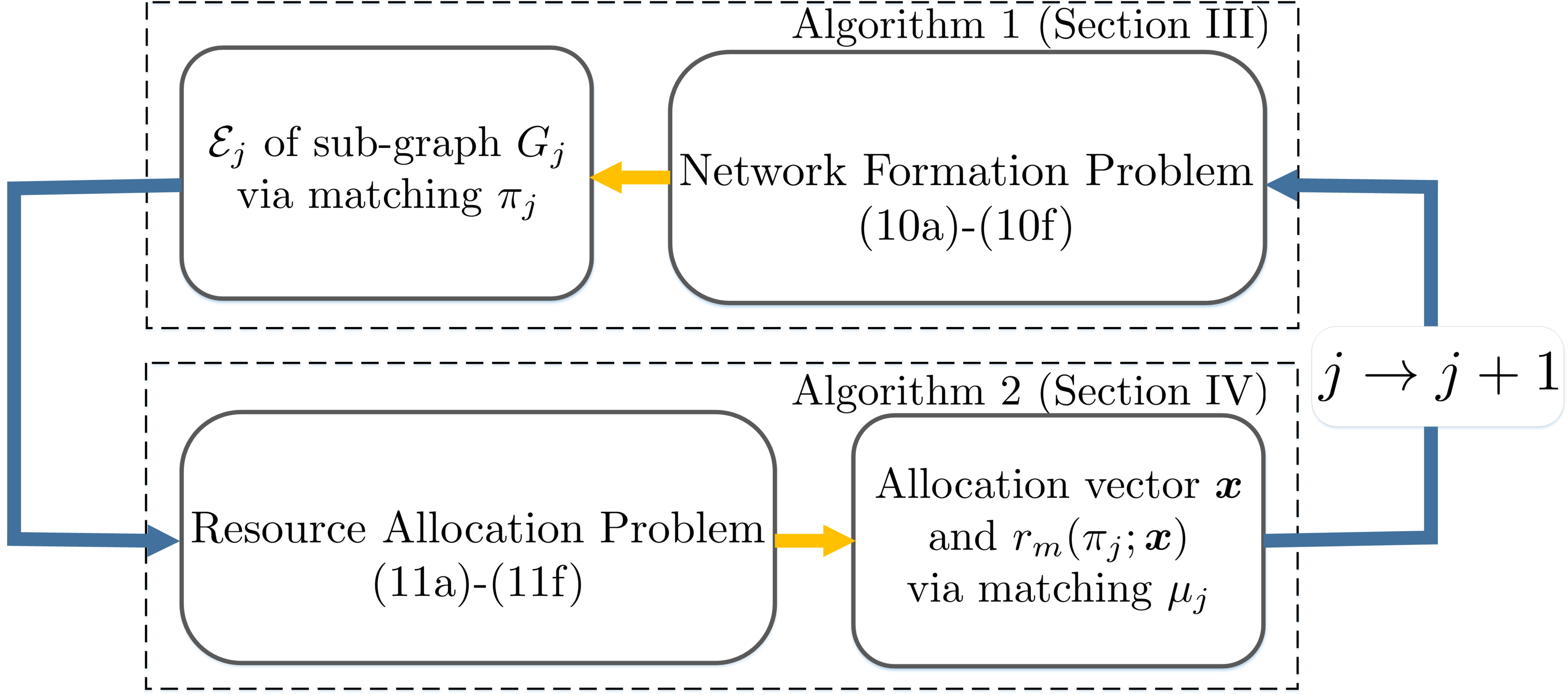}}
	\caption{\small Proposed multi-stage framework for joint backhaul network formation and resource allocation.}\vspace{0em}
	\label{organ}
\end{figure}\vspace{0em}
The solution of problem \eqref{opt1:a}-\eqref{opt1:f} yields $\mathcal{E}$ for the mmW-MBN graph $G(\mathcal{M},\mathcal{E})$ which also determines $\mathcal{M}_{m'}^{\textrm{D-BS}}$ for all $m' \in \mathcal{M}$. Next, the sub-channels of each A-BS $m'$ must be allocated to its assigned D-BSs in $\mathcal{M}_{m'}^{\textrm{D-BS}}$. Each MNO $n$ seeks to minimize the cost of its backhaul network, while maximizing the rate for each one of its SBSs $m \in \mathcal{M}_n$. To this end, the cooperative backhaul resource allocation problem can be formulated at each D-BS $m$ with $\epsilon_e(m',m)=1$ as follows:
\begin{subequations}
\begin{IEEEeqnarray}{rCll}
\!\!\!\!\!\!\!\!\!\!\!\!\argmax_{\boldsymbol{x}}&\quad&\sum_{k\in\mathcal{K}}\left[r_{m}(k,m';\zeta)-\kappa_{m}q_{m'}\mathds{1}_{m'm}\right]x_{m'km},\label{opt2:a}\\
\text{s.t.}\,\,\,\,\,\,
&&r_m(m';\boldsymbol{x})\leq \notag\\  &&\frac{1}{|\mathcal{M}_{m'}^{\textrm{D-BS}}|+1}r_{m'}(m'';\boldsymbol{x}),\hfill m' \in \mathcal{M}_{m''}^{\textrm{D-BS}}, \label{opt2:b}\\
&& r_m(m';\boldsymbol{x})\geq r_{m,\textrm{th}},\label{opt2:c}\\
&&\sum_{k \in \mathcal{K}}x_{m'km}\leq K,\label{opt2:d}\\
&&\!\!\!\!\!\!\!\!\sum_{\,\,\,\,\,\,\,\,\,\,\,m \in \pi(m')}\!\!\!\!x_{m'km}\leq 1,\label{opt2:e}\\
&&x_{m'km}\in \{0,1\}\label{opt2:f},
\end{IEEEeqnarray}
\end{subequations}
where $|.|$ denotes the set cardinality and $r_{m,\textrm{th}}$ denotes the minimum required rate for SBS $m$ which is typically determined by the traffic that is circulating over the downlink of the radio access network. \textcolor{black}{Constraint \eqref{opt2:b} ensures that the backhaul capacity of each A-BS $m'$ will be shared between its all assigned D-BSs in $\mathcal{M}_{m'}^{\text{D-BS}}$ as well as $m'$'s traffic. That is why $|\mathcal{M}_{m'}^{\text{D-BS}}|$ is increased by one in \eqref{opt2:b}. This scheme allows every A-BS receiving its traffic from the core network in addition to the traffic of the associated D-BSs.}

Prior to solving the proposed resource management problem, in \eqref{opt1:a}-\eqref{opt1:f} and \eqref{opt2:a}-\eqref{opt2:f}, we note that the solution of network formation problem will depend on the resource allocation and vice versa. That is because for any multi-hop backhaul connection, the rate of a backhaul link $e(m',m)$, $r_m(m',\boldsymbol{x})$, depends on the network formation $\mathcal{M}_{m'}^{\textrm{D-BS}}$, as shown in \eqref{opt2:b}. Moreover, to associate a D-BS $m$ to an A-BS in $\mathcal{M}_m^d$, the backhaul rates for A-BSs must be considered. Following, we propose a novel approach that allows to jointly solve these two problems\footnote{\textcolor{black}{We note that the Dijkstra's algorithm cannot be applied to the network formation and resource allocation problems at hand, since the convergence of this algorithm is contingent upon assuming a constant weight for each link, which must be independent of the weights at other links \cite{Dijkstra}. However, due to the intererence term in \eqref{rate}, this assumption will not be valid.}}. 

\section{Matching Theory for Multi-Hop Backhaul Network Formation}\label{sec:III}
The problems in \eqref{opt1:a}-\eqref{opt1:f} and \eqref{opt2:a}-\eqref{opt2:f} are $0$-$1$ integer programming which do not admit closed-form solutions and have exponential complexity \cite{4036195}. To solve these problems, we propose a novel approach based on \emph{matching theory}, a suitable mathematical framework that allows the derivation of a decentralized solution with tractable complexity for combinatorial allocation problems as shown in \cite{Roth92}, \cite{eduard11}, and \cite{Shalash2015}.

In particular, a matching game is essentially a two-sided assignment problem between two disjoint sets of players in which the players of one set must be matched to the players of the other set, according to some \textit{preference profiles}. A preference profile $\succ$ is defined as a reflexive, complete, and transitive binary relation between the elements of a given set. We denote by $\succ_m$ the preference profile of player $m$. Consequently, $a\succ_m b$ means player $m$ prefers $a$ more than~$b$.

To jointly solve the network formation and resource allocation problems, we propose a multi-stage framework, as shown in Fig. \ref{organ}, using which
the mmW-MBN can be formed as follows:
\begin{align}\label{graph}
	G_1(\mathcal{A}_1 \cup \mathcal{D}_1,\mathcal{E}_1)&\rightarrow G_2(\mathcal{A}_2 \cup \mathcal{D}_2,\mathcal{E}_2)\notag\\
	&\rightarrow \cdots \rightarrow G_J(\mathcal{A}_J \cup \mathcal{D}_J,\mathcal{E}_J),
\end{align}
where the arrows in \eqref{graph} indicate the transformation from sub-graph $G_j$ to $G_{j+1}$, where $\mathcal{A}_{j+1}=\mathcal{D}_j$ and $\mathcal{D}_{j+1} =\left\{m \in \bigcup_{m' \in \mathcal{A}_{j+1}}\mathcal{M}_{m'}^d\big|m \notin \bigcup_{j'=1}^{j+1} A_{j'}\right\}$. Each sub-graph $G_j(\mathcal{A}_j \cup \mathcal{D}_j,\mathcal{E}_j)$ is defined as a directed graph from the set of A-BSs $\mathcal{A}_j$ to the set of D-BSs $\mathcal{D}_j$ via directed links in $\mathcal{E}_j$. Initially, $\mathcal{A}_1=\{m_0\}$, and $\mathcal{D}_1=\mathcal{M}_{m_0}^d$. Each stage $j$ corresponds to the formation and resource management of $j$-th hop of the backhaul links. In fact, at each stage $j$, we address the following two problems: 1) in Subsections \ref{Sec:III-A}-\ref{Sec:III-B}, we find $\mathcal{E}_j$ of sub-graph $G_j$ that solves problem \eqref{opt1:a}-\eqref{opt1:f}, given the rate of each backhaul link from the previous stages, and 2) in Section \ref{Sec:IV}, we solve \eqref{opt2:a}-\eqref{opt2:f} for sub-graph $G_j$ to allocate the sub-channels of each A-BS in $\mathcal{A}_j$ to its associated D-BSs. The variable $J$, resulting from the proposed solution, will yield the maximum number of hops for the multi-hop backhaul link from the MBS to SBSs. The final graph $G(\mathcal{M},\mathcal{E}^*)$ is the overlay of all sub-graphs in \eqref{graph}, such that $\mathcal{E}^*=\bigcup_{j=1}^J \mathcal{E}_j$.  

\subsection{Multi-hop backhaul network formation problem as a matching game}\label{Sec:III-A}

At each stage $j$, the backhaul network formation problem can be cast as a one-to-many matching game \cite{7105641} which is defined next.\vspace{0em} 
\begin{definition}\label{def:1}
	Given two disjoint sets $\mathcal{A}_j$ and $\mathcal{D}_j$, the network formation policy $\pi_j$ can be defined as a \emph{one-to-many matching relation}, $\pi_j: \mathcal{A}_j \cup \mathcal{D}_j \rightarrow \mathcal{A}_j \cup \mathcal{D}_j$, such that
	\begin{itemize}
		\item[1)] $\forall m \in \mathcal{D}_j$, if $\pi_j(m)\neq m$, then $\pi_j(m) \in \mathcal{A}_j$,
		\item[2)] $\forall m' \in \mathcal{A}_j$, if $\pi_j(m')\neq m'$, then $\pi_j(m') \subseteq \mathcal{D}_j$,
		\item[3)] $\pi_j(m)=m'$, if and only if $m \in \pi_j(m')$,
		\item[4)] $\forall m' \in \mathcal{A}_j$, $|\pi_j(m')|\leq Q_{m'}$,
	\end{itemize}
	where $\pi_j(m)=m$ indicates that SBS $m$ is unmatched.
\end{definition}The quota of A-BS $m'$, $Q_{m'}$, represents the maximum number of D-BSs that can be assigned to $m'$. The relationship of the matching $\pi_j$ with the link formation $\mathcal{E}_j$ is such that $m \in \pi_j(m')$ is equivalent to $\epsilon_e(m',m)=1$. In addition, the matching policy $\pi_j$ by definition satisfies the constraints in \eqref{opt1:c}-\eqref{opt1:f}. 

To complete the definition of the matching game, we must introduce suitable utility functions that will subsequently be used to define the preference profiles of all players. In the proposed mmW-MBN, in addition to the achievable rate, the cost of cooperation among MNOs must be considered in the preference relations of the SBSs. Here, we define the utility of D-BS $m \in \mathcal{D}_j$ that seeks to evaluate a potential connection to an A-BS $m' \in \mathcal{A}_j$, $U_m(m')$, as
\begin{align}\label{utility1}
	U_m(m') = \min\left(\bar{r}_m(m'), r_{m'}(\pi_{j-1}(m'),\boldsymbol{x})\right)-\kappa_{m}q_{m'}\mathds{1}_{mm'},
\end{align}
where $\bar{r}_m(m')$ is given by \eqref{rate_ave_7}. Here, we note that $r_{m'}(\pi_{j-1}(m'),\boldsymbol{x})$ is determined at stage $j-1$. If $j=1$, then SBS $m$ is directly connected to the MBS. The first term in \eqref{utility1} captures the fact that achievable rate for D-BS $m$ is bounded by the backhaul rate of A-BS $m'$. The second term indicates that D-BS $m\in \mathcal{M}_n$ considers the cost of the backhaul link, if the A-BS does not belong to MNO $n$. However, if the A-BS belongs to MNO $n$, the cost will naturally be zero.

Furthermore, the utility of an A-BS $m' \in \mathcal{A}_j$ that evaluates the possibility of serving a D-BS $m \in \mathcal{D}_j$, $V_{m'}(m)$ will be:
\begin{align}\label{utility2}
V_{m'}(m) = \bar{r}_{m'}(m)+\kappa_{m'}q_{m'}\mathds{1}_{mm'}.
\end{align}

In fact, \eqref{utility2} implies that A-BS $m'$ aims to maximize the backhaul rate, while considering the revenue of providing backhaul support, if $\mathds{1}_{mm'} \neq 0$. Based on the utilities in \eqref{utility1} and \eqref{utility2}, the preference profiles of D-BSs and A-BSs will be given by:
\begin{align}
m'_1 \succ_{m}^{D} m'_2 &\iff U_{m}(m'_1) > U_{m}(m'_2),\label{pref1}\\
m_1 \succ_{m'}^{A} m_2 &\iff V_{m'}(m_1) > V_{m'}(m_2),\label{pref2}
\end{align}
where $\succ^{D}$ and $\succ^{A}$ denote, respectively, the preference relations for D-BSs and A-BSs. \textcolor{black}{Here, we assume that if $U_{m}(m')\leq 0$, A-BS $m' \in \mathcal{M}_{n'}$ will not be acceptable to D-BS $m \in \mathcal{M}_{n}$. This allows MNO $n$ to choose the control parameter $\kappa_m$ in \eqref{utility1}, to prevent the formation of any link between a given D-BS $m$ and any A-BS $m'$ that is charging a high price for using its sub-channels.} Given this formulation, we next propose an algorithmic solution for the proposed matching game that will allow finding suitable network formation policies.
\subsection{Proposed mmW-MBN formation algorithm}\label{Sec:III-B}
To solve the formulated game and find the suitable network formation policy $\pi_j$ for stage $j$, we consider two important concepts: \emph{two-sided stability} and \emph{Pareto optimality}. A two-sided stable matching is essentially a solution concept that can be used to characterize the outcome of a matching game. In particular, two-sided stability is defined as follows \cite{Roth92}:
\vspace{0em}
\begin{definition}\label{def:2}
	A pair of D-BS $m\in \mathcal{D}_j$ and A-BS $m' \in \mathcal{A}_j$ in network formation policy $\pi_j$, $(m',m)\in \pi_j$, is a \emph{blocking pair}, if and only if $m' \succ_{m}^{D} \pi_j(m)$ and $m \succ_{m'}^{A} m''$ for some $m'' \in \pi_j(m')$. A matching policy $\pi_j$ is said to be \emph{two-sided stable}, if there is no blocking pair. 
\end{definition}

The notion of two-sided stability ensures fairness for the SBSs. That is, if a D-BS $m$ prefers the assignment of another D-BS $m''$, then $m''$ must be preferred by the A-BS $\pi_j(m'')$ to $m$, otherwise, $\pi_j$ will not be two-sided stable. While two-sided stability characterizes the stability and fairness of a matching problem, the notion of Pareto optimality, defined next, can characterize the efficiency of the solution.
\vspace{0em}
\begin{definition}\label{def:3}
	A matching policy $\pi_j$ is said to be \emph{Pareto optimal} (PO), if there is no other matching $\pi'_j$ such that $\pi'_j$ is equally preferred to $\pi_j$ by all D-BSs, $\pi'_j(m) \succeq_m^D \pi_j(m)$, $\forall m \in \mathcal{D}_j$, and strictly preferred over $\pi_j$, $\pi'_j(m) \succ_m^D \pi_j(m)$ for some D-BSs.
\end{definition}
\begin{algorithm}[t!]
	\footnotesize
	\caption{Millimeter-Wave Mesh Backhaul Network Formation Algorithm}\label{euclid}
	\textbf{Inputs:}\,\,$\mathcal{A}_j$, $\mathcal{D}_j$, $\succ_{m'}^{A}$, $\succ_{m}^{D}$.\\
	\textbf{Output:}\,\,$\pi_j$.
	\begin{algorithmic}[1]
		\State \parbox[t]{.9\linewidth}{\textit{Initialize:} Temporary set of the rejected D-BSs $\mathcal{D}^r=\mathcal{D}_j$. Tentative set $\mathcal{A}_{m'}^a=\emptyset$ of accepted D-BSs by A-BS $m'$, $\forall m' \in \mathcal{A}_j$. Let $\mathcal{S}_m = \mathcal{A}_j \cap \mathcal{M}_m^d$, $\forall m \in \mathcal{D}_j$.}
		\While{$\mathcal{D}^r \neq \emptyset$}
		\State \parbox[t]{.9\linewidth}{For each D-BS $m \in \mathcal{D}^r$, find the most preferred A-BS, $m'^* \in \mathcal{S}_m$, based on $\succ_{m}^{D}$. Each D-BS $m$ sends a link request signal to its corresponding $m'^*$.} 
		\State Add $m$ to $\mathcal{A}_{m'^*}^a$ and remove $m'^*$ from $\mathcal{S}_m$. If $\mathcal{S}_m = \emptyset$, remove $m$ from $\mathcal{D}^r$.
		\State \parbox[t]{.9\linewidth}{Each A-BS $m' \in \mathcal{A}_j$ receives the proposals, tentatively accepts $Q_{m'}$ of the most preferred applicants from $\mathcal{A}_{m'}^a$, based on $\succ_{m'}^{A}$ and reject the rest.} 
		\State \parbox[t]{.9\linewidth}{Remove rejected D-BSs from $\mathcal{A}_{m'}^a$ for every A-BS $m'$ and add them to $\mathcal{D}^r$. Remove accepted D-BSs from $\mathcal{D}^r$.}
		\EndWhile
	\end{algorithmic}\label{algo:1}
\end{algorithm}

To find the stable policy $\pi_j$, the \emph{deferred acceptance} (DA) algorithm, originally introduced in \cite{Gale}, can be adopted. Hence, we introduce Algorithm \ref{algo:1} based on the DA algorithm which proceeds as follows. Initially, no D-BS in $\mathcal{D}_j$ is assigned to an A-BS in $\mathcal{A}_j$. The algorithm starts by D-BSs sending a link request signal to their most preferred A-BS, based on their preference relation $\succ_{m}^{D}$. Next, each A-BS $m'$ receives the request signals and approves up to $Q_{m'}$ of the most preferred D-BSs, based on $\succ_{m'}^{A}$ and rejects the rest of the applicants. The algorithm follows by rejected D-BSs applying for their next most preferred A-BS. Algorithm \ref{algo:1} converges once each D-BS $m$ is assigned to an A-BS or is rejected by all A-BS in $\mathcal{A}_j \cap \mathcal{M}_m^d$. Since it is based on a variant of the DA process, Algorithm \ref{algo:1} is guaranteed to converge to a stable matching as shown in \cite{Gale}. Moreover, among the set of all stable solutions, Algorithm \ref{algo:1} yields the solution that is PO for the D-BSs. \textcolor{black}{Here, we note that the role of an SBS will change dynamically according to the changes of the CSI. However, due to the slow-varying channels, the CSI will remain relatively static within the channel coherence time (CCT), and consequently, the role of SBSs can be considered fixed within one CCT. The proposed distributed solution in Algorithm \ref{algo:1} allows the SBSs to update their preference profiles, which depend on the CSI, and accordingly their role, after each CCT period.}

Given $\pi_j$ resulted from Algorithm \ref{algo:1}, the sub-channels of each A-BS $m'\in \mathcal{A}_j$ must be allocated to the D-BSs in $\pi_j(m')$. To this end, we next propose a distributed solution to solve the backhaul resource allocation problem.
\section{Matching Theory for Distributed  Backhaul Resource Management}\label{Sec:IV}
To solve the problem in \eqref{opt2:a}-\eqref{opt2:f}, centralized approaches will require MNOs to share the information from their SBSs with a trusted control center. Therefore, centralized approaches will not be practical to perform inter-operator resource management. To this end, we formulate the problem in \eqref{opt2:a}-\eqref{opt2:f} in each stage $j$ as a second matching game and propose a novel distributed algorithm to solve the problem. The resulting resource allocation over stage $j$ will 
determine the rate for each backhaul link in $j$-th hop. As discussed in the previous section,
this information will be used to find the network formation policy $\pi_{j+1}$ of the next stage $j+1$.
\vspace{-1em}
\subsection{Resource management as a matching game}
To allocate the sub-channels of an A-BS $m' \in \mathcal{A}_j$ to its associated D-BSs in $\pi_j(m')$, we consider a one-to-many matching game $\mu_j$ composed of two disjoint sets of mmW sub-channels, $\mathcal{K}$, and the D-BSs in $\pi_j(m')$ associated to A-BS $m'$. The matching $\mu_j$ can be formally defined, similar to the network formation matching $\pi$ in Definition \ref{def:1}.
However, unlike in the network formation matching game, here, we do not introduce any quota for the D-BSs. There are two key reasons for not considering quotas in our problem which can be explained as follows. First, for a resource allocation problem with minimum rate requirement, as presented in \eqref{opt2:b} and \eqref{opt2:c}, a quota cannot be determined a priori for an SBS. That is because the number of sub-channels required by a given SBS is a function of the CSI over all sub-channels between an A-BS $m'$ and all D-BSs assigned to $m'$. Therefore, sub-channel allocation for one D-BS will affect the number of required sub-channels by other D-BSs. This is a significant difference from classical solutions based on matching theory such as in \cite{Shalash2015,eduard11,7105641}, and \cite{7268835}. The second practical reason for not using a fixed quota in our resource allocation problem is that there is no clear approach to determine the suitable quota values as a function of the various system metrics, such as CSI. On the other hand, with no constraint on the maximum number of sub-channels to be allocated to a D-BS, a distributed matching algorithm may assign all the sub-channels to a few of D-BSs, resulting in an inefficient allocation. Therefore, considering the significant impact of quota on the resource allocation, it is more practical to limit the number of allocated sub-channels by the natural constraint of the system, as presented in constraint \eqref{opt2:b}.

Therefore, we let a D-BS $m$ assign a utility $\Psi_m(k;\mu_j)$ to a sub-channel $k$, where $(m,k)\notin \mu_j$, only if
\begin{align}\label{criterion}
\sum_{k' \in \mu_j(m)}r_m(k',\pi_j(m))<\! \frac{1}{|\mathcal{M}_{\pi_{j}(m)}^{\textrm{D-BS}}|+1}r_{\pi_{j}(m)}(m'';\boldsymbol{x}),
\end{align}
where $m''$ is the A-BS that serves $\pi_j(m)$, i.e., $m''=\pi_{j-1}(\pi_j(m))$. 
Otherwise, $\Psi_{m}(k;\mu_j) = -\infty$, meaning that sub-channel $k$ is not acceptable to D-BS $m$, given the current matching $\mu_j$. In fact, \eqref{criterion} follows the rate constraint in \eqref{opt2:b} and prevents the D-BS $m$ from being allocated to unnecessary sub-channels. 
With this in mind, we define the utilities and preferences of sub-channels and D-BSs, considering the rate constraints, CSI, and the cost of each sub-channel. For any D-BS $m$, the utility that $m$ achieves when being matched to a sub-channel $k$ will be given~by:
\begin{align}\label{utility:3}
\Psi_m(k;\mu_j)\! = \!\!\begin{cases}
r_m(k,\pi_j(m)), &\textrm{if} \,\, \eqref{criterion}\,\,\textrm{is held}  ,\\ 
-\infty, &\textrm{otherwise.}
\end{cases}
\end{align}
Here, note that \eqref{utility:3} does not include the price of the sub-channels. That is because the sub-channel price $q_{\pi_j(m)}$ is equal for all sub-channels and will not affect the preference of the D-BS.

The utility of sub-channels is controlled by their corresponding A-BS. The utility that is achieved by sub-channel $k$ when being matched to a D-BS $m$ will be:
\begin{align}\label{utility:4}
\Phi_k(m) = r_m(k,\pi_j(m)) + \kappa_{\pi_j(m)}q_{\pi_j(m)}\mathds{1}_{m\pi_j(m)}.
\end{align}

In \eqref{utility:4}, the second term represents the revenue obtained by A-BS $\pi_j(m)$ for providing backhaul support to D-BS $m$ over sub-channel $k$. The scaling factor $\kappa_{\pi_j(m)}$ enables the A-BS to balance between the achievable rate and the revenue. In fact, as $\kappa_{\pi_j(m)}$ increases, a given A-BS will tend to assign more resources to D-BSs of other MNOs. Similar to \eqref{pref1} and \eqref{pref2}, the preference profiles of sub-channels and D-BSs are given by 
\begin{align}
k_1 P_{m}^{D} k_2 &\iff \Psi_{m}(k_1) > \Psi_{m}(k_2),\label{pref3}\\
m_1 P_{k}^{K} m_2 &\iff \Phi_{k}(m_1) > \Phi_{k}(m_2),\label{pref4}
\end{align}
where $P_{m}^{D}$ and $P_{k}^{K}$ denote, respectively, the preference profiles of D-BS $m$ and sub-channel $k$.\vspace{0em}
\subsection{Proposed resource allocation algorithm for mmW-MBNs}
Here, our goal is to find a two-sided stable and efficient PO matching $\mu_j$ between the sub-channels of A-BS $m'$ and the D-BSs in $\pi_j(m')$, for every A-BS $m' \in \mathcal{A}_j$. From \eqref{utility:3}, we observe that the utility of a D-BS and, consequently, its preference ordering depend on the matching of the other D-BSs. This type of game is known as a \textit{matching game with peer effect} \cite{Bodine11}. 
This is in contrast
with the traditional matching games in which players have strict and non-varying preference profiles. For the proposed matching game, we can make the following observation.
\vspace{0cm}
\begin{proposition}
Under the mmW-MBN specific utility functions in \eqref{utility:3} and \eqref{utility:4}, the conventional DA algorithm is not guaranteed to yield a stable solution.
\end{proposition}
\begin{proof}\vspace{0em}
We prove this using an example. Let  $\mathcal{D}_j=\{m_1,m_2\}$, with	preference profiles $k_1 P_{m_1}^D k_2 P_{m_1}^D k_3$ and $k_2 P_{m_2}^D k_3 P_{m_2}^D k_1$. Moreover, let $\mathcal{K}=\{k_1,k_2,k_3\}$ with preference profiles $m_1 P_{k_i}^K m_2$, for $i=1,2$, and $m_2 P_{k_3}^K m_1$. Considering achievable rates $r_{m_1}(k_1)$, $r_{m_2}(k_2)$ and $r_{m_2}(k_3)$ are greater than $r_{\textrm{th}}$, DA algorithm yields a matching $\mu$ where $\mu(k_1)=\{m_1\}$, $\mu(k_2)=\emptyset$, and $\mu(k_3)=\{m_2\}$. However, $(m_2,k_2)$ form a blocking pair, which means $\mu$ is not stable. 
\end{proof}

Thus, we cannot directly apply the DA algorithm to our problem and we need to adopt a novel algorithm that handles blocking pairs and achieves a stable solution. To this end, we proposed a distributed resource allocation scheme in Algorithm \ref{algo:2}. The algorithm proceeds as follows. For every A-BS $m' \in \mathcal{A}_j$, the initial set of rejected sub-channels is $\mathcal{K}^r=\mathcal{K}$. The algorithm initiates by each sub-channel $k \in \mathcal{K}^r$ sending a request signal to its most preferred D-BS, based on \eqref{pref4}. The D-BSs receive the requests and accept a subset of most preferred sub-channels, based on  \eqref{pref3}, that satisfy their minimum rate requirement and reject the rest. The rejected sub-channels are added to $\mathcal{K}^r$. Accepted sub-channels and the sub-channels that are rejected by all D-BSs in $\pi_j(m')$ are removed from $\mathcal{K}^r$. Moreover, in step $9$, D-BSs update their preferences $P_m^D$. The rejected sub-channels apply for their next most preferred D-BS from their preference profile. Algorithm \ref{algo:2} proceeds until $\mathcal{K}^r$ is an empty set. Next, any unmatched sub-channel $k \in \mathcal{K}$ is assigned to the most preferred D-BS $m$ with $\mathds{1}_{mm'}=0$. The algorithm converges once all sub-channels are matched. Throughout this algorithm, we note that the corresponding A-BS sends the matching requests to D-BSs on the behalf of its sub-channels. 
\begin{algorithm}[!t]
	\footnotesize
	\caption{Backhaul Resource Allocation Algorithm}\label{euclid}
	\textbf{Inputs:}\,\,$\mathcal{A}_j$, $\pi_j$, $r_{\text{th}}$.\\
	\textbf{Output:}\,\,$\mu_j$.
	\begin{algorithmic}[1]
		\For{$i=1$, $i\leq |\mathcal{A}_j|$,$i++$}
		\State \parbox[t]{.9\linewidth}{\textit{Initialize:} Set A-BS $m'$ to $i$-th element of $\mathcal{A}_j$. Temporary set of the rejected sub-channels $\mathcal{K}^r=\mathcal{K}$. Tentative sets $\mathcal{D}_{j}^m=\emptyset$ and $r_m(m')=0$ for each D-BS $m \in \pi_j(m')$. For each sub-channel $k$, let $\mathcal{C}_k=\pi_j(m')$.}
		\While{$\mathcal{K}^r \neq \emptyset$}
		\State \parbox[t]{.9\linewidth}{For each sub-channel $k \in \mathcal{K}^r$, find the most preferred D-BS, $m^* \in \pi_j(m')$, based on $P_{k}^{K}$. A-BS sends a link request signal to the corresponding $m^*$ for each sub-channel. Add $k$ to $\mathcal{D}_{j}^{m^*}$ and remove $m^*$ from $\mathcal{C}_k$.}
		\State If $\mathcal{C}_k = \emptyset$, remove $k$ from $\mathcal{K}^r$.
		\State \parbox[t]{.9\linewidth}{Each D-BS $m \in \pi_j(m')$ receives the proposals and tentatively accepts the most preferred sub-channel from $\mathcal{D}_{j}^m$, based on $P_{m}^{D}$ and adds the corresponding rate of the accepted sub-channels to $r_m(m')$. }
		\State  \parbox[t]{.9\linewidth}{If \eqref{criterion} is not met, add the next most preferred sub-channel and update $r_m(m')$, otherwise, reject the rest of sub-channels.}
		\State \parbox[t]{.9\linewidth}{Remove rejected sub-channels from $\mathcal{D}_{j}^m$ for every D-BS $m$ and add them to $\mathcal{K}^r$. Remove accepted sub-channels from $\mathcal{K}^r$.}
		\State Update $\mu_j$ and $P_{m}^D$ for every D-BS.
		\EndWhile
		\State Based on current allocation, update $\Psi_m(k;\mu_j)$ for D-BSs with $\mathds{1}_{mm'}=0$.
		\State Let $\mathcal{K}'^r = \{k \in \mathcal{K}|\mu_j(k)=k\}$.
		\While{$\mathcal{K}'^r \neq \emptyset$}
		\State \parbox[t]{.9\linewidth}{Remove an arbitrary sub-channel $k$ from $\mathcal{K}'^r$ and allocate it to its most preferred D-BS $m$ from $\pi_j(m')$, with $\mathds{1}_{mm'}=0$, if \eqref{criterion} is held.}
		\State Update $\mu_j$, $P_m^D$, and $r_m(m')$. 
		\EndWhile
		\EndFor
	\end{algorithmic}\label{algo:2}
\end{algorithm}
The proposed Algorithm \ref{algo:2} exhibits the following properties. 
\begin{theorem}
	Algorithm \ref{algo:2} is guaranteed to converge to a two-sided stable matching $\mu_j$ between sub-channels and D-BSs. Moreover, the resulting solution, among all possible stable matchings, is Pareto optimal for sub-channels.
\end{theorem} 
\begin{proof}\vspace{0em}
	Algorithm \ref{algo:2} will always converge, since no sub-channel will apply for the same D-BS more than once, through steps $3$-$12$ or $13$-$16$. Next, we show that the proposed algorithm always converges to a two-sided stable matching. To this end, let D-BS $m$ and sub-channel $k$ form a blocking pair $(m,k)$. That is, $m P_k^K \mu_j(k)$ and $k P_m^D k'$, where $k' \in \mu_j(m)$. We show that such a blocking pair does not exist. To this end, we note that there are two possible cases for sub-channel $k$: 1) $k \in \mathcal{K}'^r$, and 2) $k \notin \mathcal{K}'^r$, in step $12$. 
	
	If $k \in \mathcal{K}'^r$ in step $12$, that means $k$ is unmatched and the minimum rate requirement is satisfied for all D-BSs, including $m$. Thus, $(m,k)$ is a blocking pair only if $\mathds{1}_{mm'}=0$, otherwise, $m$ will not accept more sub-channels.  However, in step $14$, $k$ will be assigned to its most preferred D-BS $\mu_j(k)$ with $\mathds{1}_{\mu_j(k)m'}=0$, meaning that $\mu_j(k) P_k^K m$. Hence, $(m,k)$ cannot be a blocking pair.
	
	Next, if $k \notin \mathcal{K}'^r$, then $k$ is matched to a D-BS $\mu_j(k)$ prior to step $12$. Here, $m P_k^K \mu_j(k)$ implies that $k$ has applied for $m$ before $\mu_j(k)$ and is rejected. Thus, $k' P_m^D k$, for all $k' \in \mu_j(m)$. Therefore, $(m,k)$ cannot be a blocking pair and matching $\mu_j$ is two-sided stable. 

	
	To prove Pareto optimality, we show that no sub-channel $k$ can improve its utility by being assigned to another D-BS $m$, instead of $\mu_j(k)$. If $m P_k^K \mu_j(k)$, it means  $k' P_m^D k$, for all $k' \in \mu_j(m)$, due to the two-sided stability of $\mu_j$. Hence, a new matching $\mu'_j$ that allocates $k$ to $m$ instead of a sub-channel $k' \in \mu_j(m)$ will make $(m,k')$ to be a blocking pair for $\mu'_j$. Therefore, no other stable matching exists that improves the utility of a sub-channel.
\end{proof} 
Here, we note that the proposed solution is Pareto optimal within each subgraph, corresponding to each stage, and it is assumed that network formation in subsequent subgraphs will not affect the utility functions in \eqref{utility1} and \eqref{utility2} for the SBSs in previous subgraphs. This assumption is valid, since a given D-BS will experience random interference from the interfering A-BSs. Given that the number of interfering A-BSs is large, which is true for backhaul networks that are supporting many SBSs, the average interference power in \eqref{utility1} and \eqref{utility2} will not depend on the network formation in subsequent subgraphs. Hence, the preference profiles of the D-BSs and A-BSs, and the consequent matching within each subgraph will be independent of the other subgraphs. Therefore, given that matching within each subgraph is Pareto optimal and is not affected by other subgraphs, the overall network formation is Pareto optimal in terms of maximizing the sum-rate.
\vspace{-0cm}

\subsection{Complexity Analysis of the Proposed Multi-stage Solution}
First, we analyze the network formation complexity of an arbitrary stage $j$ from Algorithm \ref{algo:1}. For the purpose of complexity analysis, we consider the maximum number of requesting signals that D-BSs in $\mathcal{D}_j$ will send to the A-BSs in $\mathcal{A}_j$ before Algorithm  \ref{algo:1} converges. In the worst-case scenario, i.e., the scenario with the highest conflict among D-BSs, all D-BSs in $\mathcal{D}_j$ have the same preference ordering. Let 
\begin{align}\label{complex1}
m'_{1}	\succ_{m} m'_{2} \succ_{m} \cdots \succ_{m} m'_{|\mathcal{A}_j|-1} \succ_{m} m'_{|\mathcal{A}_j|},  
\end{align}
be the preference ordering of all D-BSs $m \in \mathcal{D}_j$, where $\mathcal{A}_j=\{m'_{1},m'_{2},\cdots,m'_{|\mathcal{A}_j|}\}$. Hence, only $Q_{m'_i}$ D-BSs will be accepted by the A-BS $m'_i$ during the $i$-th iteration of Algorithm \ref{algo:1}. Moreover, the number of iterations $I$ is an integer that satisfies
\begin{align}\label{complex2}
\sum_{i=1}^{I-1}Q_{m'_i} < |\mathcal{D}_j| \leq \sum_{i=1}^{I}Q_{m'_i}.
\end{align}
Therefore, the total number of requests sent by D-BSs will be
\begin{align}\label{complex3}
&|\mathcal{D}_j|+ \left(|\mathcal{D}_j|-Q_{m'_1}\right) + \left(|\mathcal{D}_j|-Q_{m'_1}-Q_{m'_2}\right) + \cdots, \notag\\
&
+\left(|\mathcal{D}_j|-Q_{m'_1}-\cdots-Q_{m'_{I-1}}\right) =I|\mathcal{D}_j| -\sum_{i=1}^{I-1}(I-i)Q_{m'_i},\notag\\
&=I|\mathcal{D}_j| -I\sum_{i=1}^{I}Q_{m'_i}+\sum_{i=1}^{I}iQ_{m'_i}\leq \sum_{i=1}^{I}iQ_{m'_i},
\end{align}
where \eqref{complex2} is used to derive the inequality in \eqref{complex3}. For the special case in which $Q_{m'_i}=Q, \forall m'_i \in \mathcal{A}_j$, \eqref{complex2} implies that $(I-1)Q<|\mathcal{D}_j|$. Hence, \eqref{complex3} can be simplified to
\begin{align}\label{complex4}
\sum_{i=1}^{I}iQ_{m'_i}=\frac{1}{2}Q(I)(I+1)< \frac{1}{2}Q\left(\frac{|\mathcal{D}_j|}{Q}+1\right)\left(\frac{|\mathcal{D}_j|}{Q}+2\right).
\end{align}
Therefore, the complexity of stage $j$ in Algorithm \ref{algo:1} is $\mathcal{O}(|\mathcal{D}_j|^2)$, which admits a second-order polynomial relation with respect to the number of D-BSs.\\
Similarly, for Algorithm \ref{algo:2}, the worst case scenario is when all sub-channels of A-BS $m'$ have the same preference ordering for D-BSs in $\pi_j(m')$ and only one sub-channel is accepted over each iteration. Therefore, the number of requesting signals sent from A-BS $m'$ to its associated D-BSs in $\pi_j(m')$ will be at most
\begin{align}\label{complex5}
\!\!\!\!\!\!K&+(K-1)+\cdots+\left(K-|\pi_j(m')|+1\right),\notag\\=&|\pi_j(m')|K-\frac{1}{2}\left(|\pi_j(m')|\right)\left(|\pi_j(m')|+1\right)< KQ_{m'},
\end{align}
where the inequality in \eqref{complex5} results from having $0\leq |\pi_j(m')|\leq Q_{m'}$. Therefore, the total number of requesting signals is $\sum_{m' \in \mathcal{A}_j}KQ_{m'}$. For $Q_{m'_i}=Q, \forall m'_i \in \mathcal{A}_j$, the complexity of Algorithm \ref{algo:2} in stage $j$ is $\mathcal{O}(KQ|\mathcal{A}_j|)$. Thus, the overall complexity of an arbitrary stage $j$ of the proposed distributed solution is $\mathcal{O}(|\mathcal{D}_j|^2+|\mathcal{A}_j|)$. This result implies that the complexity of the proposed distributed solution, composed of Algorithms \ref{algo:1} and \ref{algo:2}, is bounded by a second-order polynomial with respect to the network size. This result shows that the proposed approach yields a solution with a manageable complexity for the two interrelated integer programming problems in \eqref{opt1:a}-\eqref{opt1:f} and \eqref{opt2:a}-\eqref{opt2:f}. 
\begin{table}[!t]
	\scriptsize
	\centering
	\caption{
		\vspace*{-0em}Simulation parameters}\vspace*{-0em}
	\textcolor{black}{
		\begin{tabular}{|c|c|c|}
			\hline
			\bf{Notation} & \bf{Parameter} & \bf{Value} \\
			\hline
			$f_c$ & Carrier frequency & $73$ GHz\\
			\hline
			$p_{t,m_0}$, $p_{t,m}$ &  Transmit power for MBS \& SBSs & $40$ and $30$ dBm\\
			\hline
			$M$ & Total number of SBSs & $3$ to $65$\\
			\hline
			$N$ &  Number of MNOs & $3$ to $5$\\
			\hline
			$\Omega$ & Available Bandwidth & $5$ GHz\\
			\hline
			$K$ & Number of sub-channel& $50$\\
			\hline
			($\xi_{\text{LoS}}$,$\xi_{\text{NLoS}}$) & Standard deviation of path loss& ($4.2, 7.9$) \cite{Ghosh14} \\
			\hline
			($\alpha_{\text{LoS}}$,$\alpha_{\text{NLoS}}$) & Path loss exponent& ($2,3.5$) \cite{Ghosh14}\\
			\hline
			$d_0$ & Path loss reference distance& $1$ m \cite{Ghosh14}\\
			\hline
			$G_{\text{max}}$ & Antenna main lobe gain& $18$ dB\cite{7110547} \\
			\hline
			$G_{\text{min}}$ & Antenna side lobe gain& $-2$ dB\cite{7110547} \\
			\hline
			$\theta_{m}$ & beam width& $10^{\circ}$\cite{7110547} \\
			\hline
			$\sigma^2$ & Noise power& $-174$ dBm/Hz \!\!+\!\! $10\log_{10}\!\!\frac{\Omega}{K}$ \\
			\hline
			$d_{\text{max}}$ & Radius of simulation area& $400$ m  \\
			\hline
			$d$ & SBSs communication range& $200$ m \cite{7414178,rappaport2014}\\
			\hline
			$r_{\text{th}}$ & Required rate per SBS& $1$ Mbps \\
			\hline
			$q$ & Unit of price per sub-channel& $\$ 1$  \\
			\hline
		\end{tabular}\label{tab:sim}\vspace{-0cm}
	}
\end{table}
\vspace{0em}
\section{Simulation Results}\label{simulations}
For our simulations, we consider a mmW-MBN with an MBS located at $(0,0)\in \mathbbm{R}^2$ and up to $M = 65$ SBSs distributed uniformly and randomly within a planar area with radius $d_{\textrm{max}}=400$ m. The simulation parameters are summarized in Table \ref{tab:sim}. Moreover, the number of SBSs is considered to be equal for all MNOs. We compare our proposed approach with the following three other approaches:
\begin{enumerate}
	\item[1)]  \textit{Optimal solution} obtained via an exhaustive search which finds the resource allocation that maximizes the \textcolor{black}{backhaul sum rate. In fact, this benchmark explores all the possibilities for sub-channel allocation with uniform transmission power.}
	\item[2)]  \textit{Non-cooperative scheme} \textcolor{black}{which follows the proposed algorithms for both network formation and resource allocation, however, cooperation among MNOs is not allowed. That is, the SBSs of an MNO do not provide backhaul support to the SBSs of other MNOs.}
	\item[3)]  \textit{Random allocation} that assigns D-BSs randomly to an A-BS within their communication range, subject to the constraints in \eqref{opt1:b}-\eqref{opt1:f}. In addition, each A-BS randomly allocates sub-channels to its assigned D-BSs, subject to the constraints in \eqref{opt2:b}-\eqref{opt2:f}.
\end{enumerate}
   All statistical results are averaged over a large number of independent runs.
\subsection{Achievable \textcolor{black}{backhaul} sum rate of the mmW-MBN}
\begin{figure}[t!]
	\centering
	\begin{subfigure}[b]{\columnwidth}
		\includegraphics[width=\columnwidth]{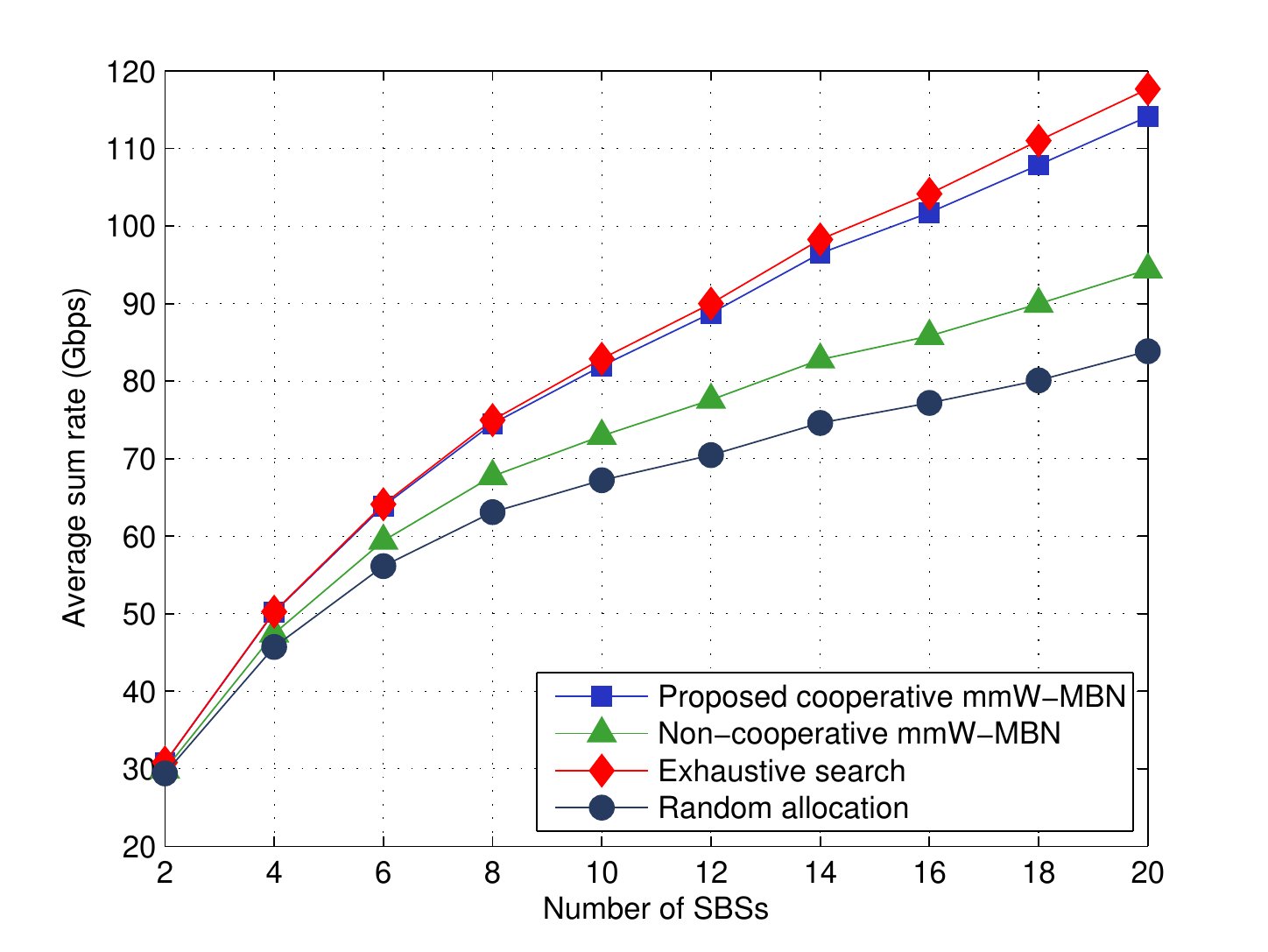}\vspace{-1em}
		\caption{}
		\label{ffig3}\vspace{-0em}
	\end{subfigure}
	\begin{subfigure}[b]{\columnwidth}
		\includegraphics[width=\columnwidth]{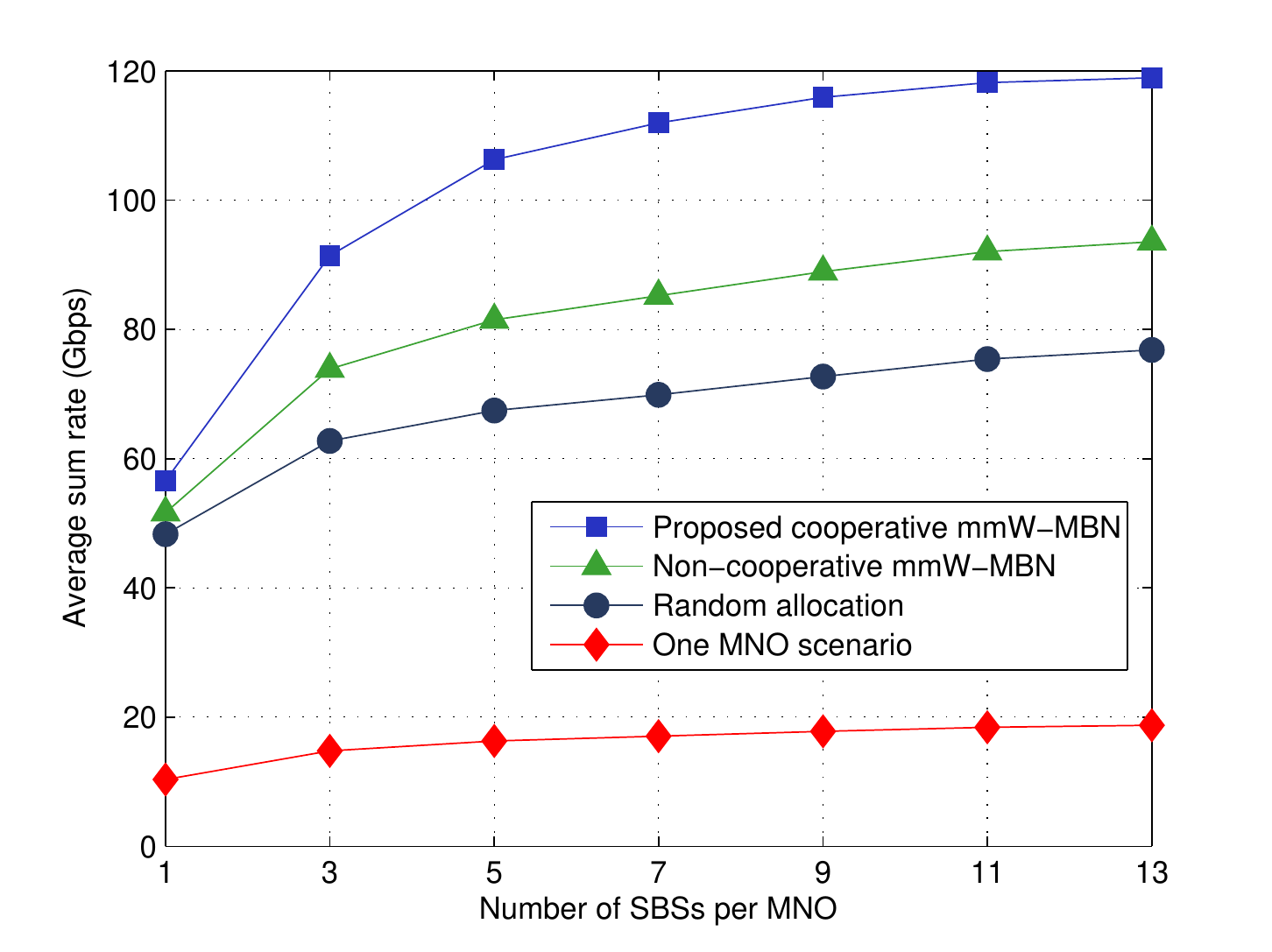}\vspace{-0.1cm}
		\caption{}
		\label{ffig4}
	\end{subfigure}
	\vspace{-2em} 
	\caption{\small Average sum rate resulting from the proposed cooperative mm-MBN approach, non-cooperative scheme, random allocation, and the optimal solution as the number of SBSs varies.}\label{figs3-4}
\end{figure}
Fig. \ref{ffig3} shows a performance comparison between the proposed framework with the optimal solution, non-cooperative, and random allocation approaches, for a mmW-MBN with $K=7$ sub-channels, up to $M=20$ SBSs, and $N=2$ MNOs. Due to the computational complexity of the exhaustive search, for this comparison figure, a relatively small network size is considered. In Fig. \ref{ffig3}, the optimal solution and the random allocation provide, respectively, an upper and lower bound on the achievable sum-rate of the given network. Fig. \ref{ffig3} shows that the proposed cooperative framework based on matching theory yields a promising performance comparable with results from the optimal solution. In fact, the performance gap will not exceed $3.2 \%$ with the network size up to $M=20$ SBSs. In addition, the results in Fig. \ref{ffig3} show that the proposed solution improves the sum-rate up to $21\%$ and $36\%$ compared to, respectively, the non-cooperative and the random allocation scheme.

\textcolor{black}{
In Fig. \ref{ffig4}, the average sum rate resulting from the proposed cooperative approach is compared with both the non-cooperative and random allocation schemes, for a dense mmW-MBN with $N=5$ MNOs and up to $M=65$ SBSs. From Fig. \ref{ffig4}, we can see that the average sum rate increases as the number of SBSs increases. This is due to the fact that more SBSs will be able to connect to the MBS via a multi-hop backhaul link. Fig. \ref{ffig4} shows that, the proposed approach outperforms both the non-cooperative and random allocation schemes for all network sizes. In fact, the proposed framework increases the average sum rate by $27 \%$ and $54 \%$, respectively, compared to the non-cooperative and random allocation schemes, for $M=65$ SBSs. \textcolor{black}{From this figure, we can clearly see that the average sum rate and, hence, the spectral efficiency of the network is significantly improved in multi-MNO scenarios compared with the network scenario with only one MNO. That is because directional transmissions over the millimeter wave frequencies allow different MNOs to efficiently reuse the available bandwidth which, in turn, results in higher spectral efficiency. } }
\begin{figure}[!t]
	\centering
	\centerline{\includegraphics[width=\columnwidth]{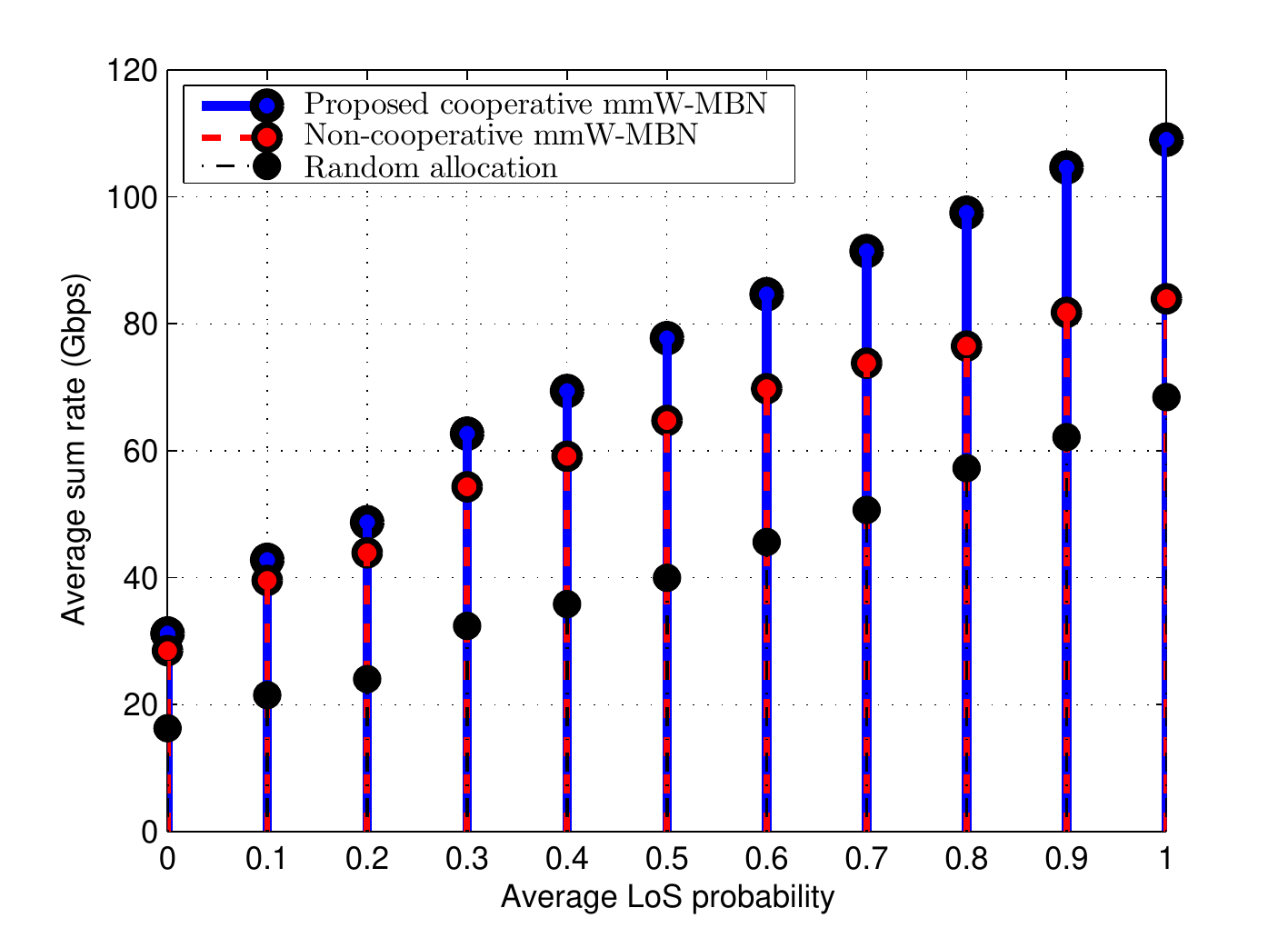}}\vspace{-.2cm}
		\caption{\small Average sum rate versus the average LoS probability $\rho$.}\vspace{0em}
	\label{ffig5}
\end{figure}\vspace{0em}

\textcolor{black}{
In Fig. \ref{ffig5}, we compare the average sum rate for the proposed approach with the non-cooperative and random allocation schemes, versus the average LoS probability. The primary goal here is to analyze the severe impact of the blockage on the network performance. In particular, Fig. \ref{ffig5} shows that blockage degrades the sum rate up to six times, when $\rho$ decreases from $1$ to $0$. However, the results in Fig. \ref{ffig5} show that the proposed approach is more robust against blockage, compared to the non-cooperative and random allocation schemes. In fact, the proposed approach yields up to $25 \%$ and $42 \%$ performance gains for $\rho=1$, respectively, compared to the non-cooperative and random allocation schemes. We note that for extreme blockage scenarios, e.g., $\rho = 0, 0.2$, it is expected that the gains will be small, since the achievable rate for most of the links is degraded by blockage. However, we can observe that as more LoS backhaul links become available, the performance gap increases. The main reason for this trend is that the backhaul rate for each A-BS increases, as $\rho$ increases, which can support higher rates for its associated D-BSs. The average sum rate increases by $27\%$ for the proposed cooperative mmW-MBN, as $\rho$ increases from $0.6$ to $1$.}
\vspace{0em}

\subsection{Statistics of the achievable rate for the mmW-MBN}
\begin{figure}[!t]
	\centering
	\centerline{\includegraphics[width=\columnwidth]{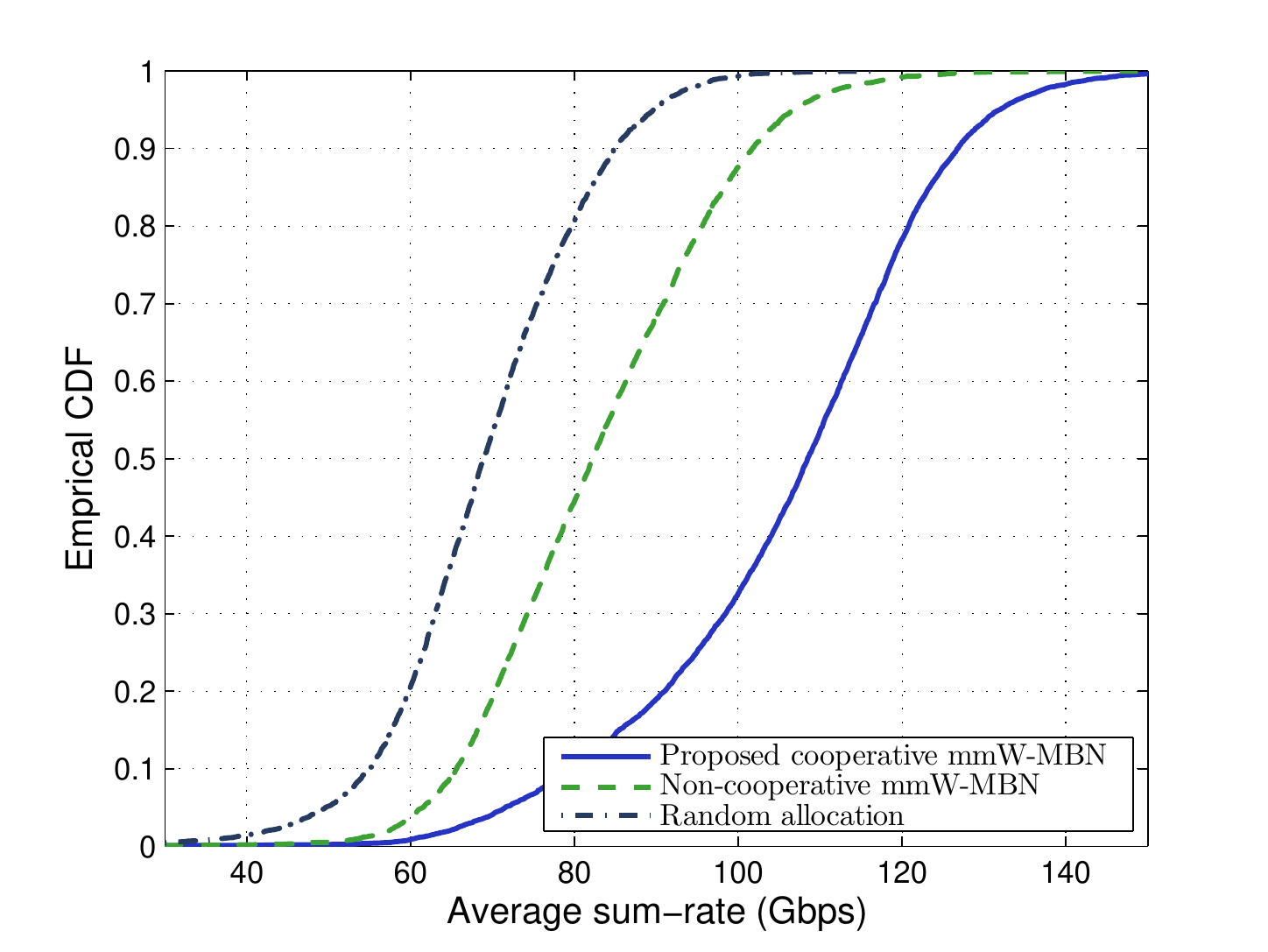}}\vspace{-.2cm}
	\caption{\small The CDF of the average sum rate resulting from the proposed cooperative mmW-MBN, the non-cooperative baseline, and the random allocation approach.}\vspace{0em}
	\label{ffig6}
\end{figure}\vspace{0em}

Fig. \ref{ffig6} shows the cumulative distribution function (CDF) of the average sum rate for the proposed cooperative approach, compared to the non-cooperative and random allocation schemes, for $N=5$ MNOs and $M=60$ SBSs. The results show that the proposed cooperative approach substantially improves the statistics of the average sum rate. For example, Fig. \ref{ffig6} shows that the probability of achieving a $80$ Gbps target sum rate is $90\%$, $36\%$, and $20\%$, respectively, for the proposed approach, the non-cooperative scheme, and the random allocation. 

\begin{figure}[!t]
	\centering
	\centerline{\includegraphics[width=\columnwidth]{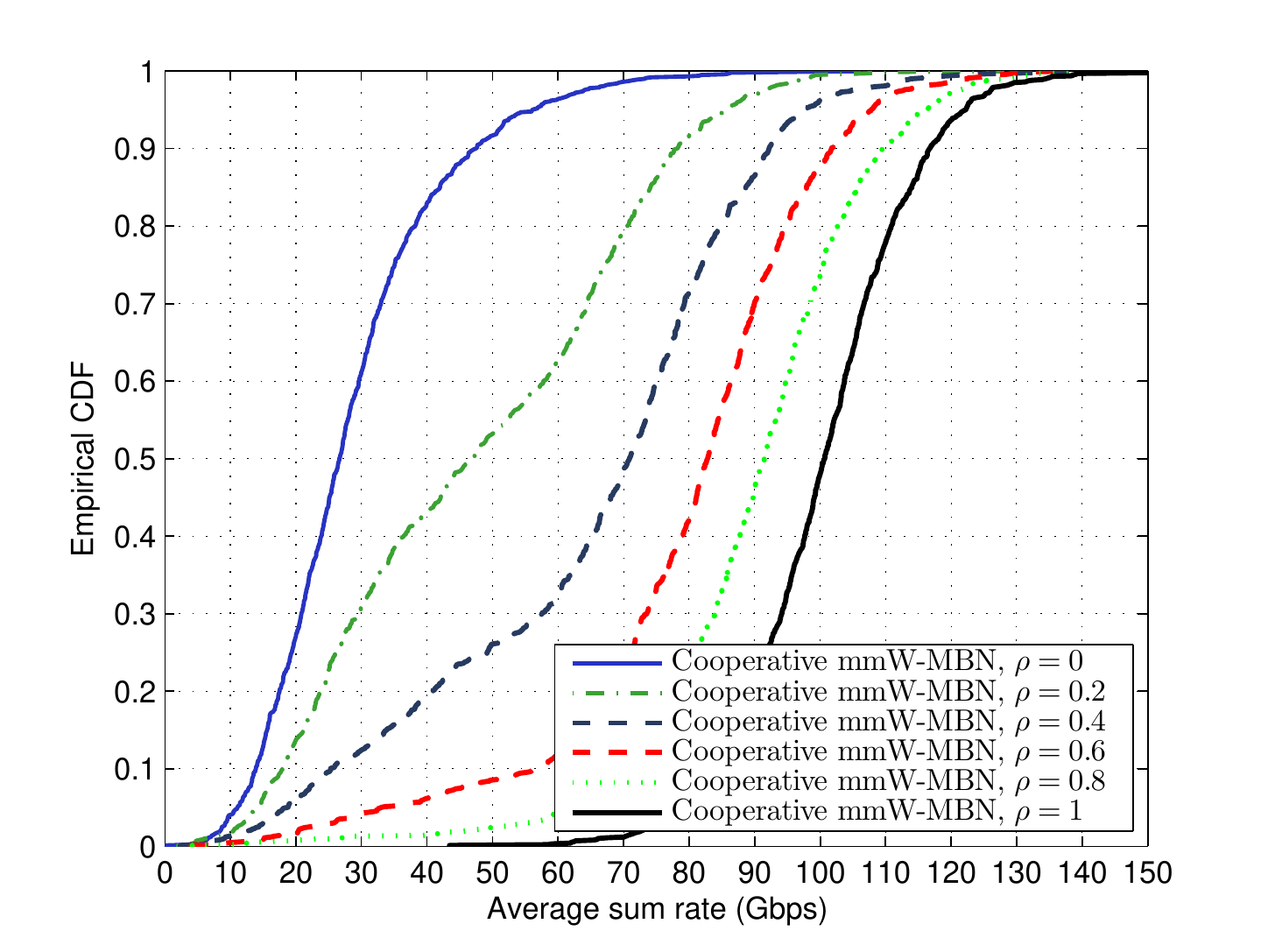}}\vspace{-.2cm}
	\caption{\small The empirical CDF of the sum rate for different average LoS probabilities.}\vspace{0em}
	\label{ffig7}
\end{figure}\vspace{0em}

\textcolor{black}{
Fig. \ref{ffig7} shows the empirical CDF of the average sum rate for different average LoS probabilities, for $N=5$ MNOs and $M=20$ SBSs. From this figure, we can see that  severe blockage with small $\rho$ significantly degrades the performance of the mmW-MBN. Interestingly, we can observe that with $\rho=1$, the average sum rate does not fall below the $40$ Gbps. However, as the probability of LoS decreases to $0.2$, the probability of the average sum rate be less than $40$ Gbps is $42 \%$.}
\textcolor{black}{
\subsection{Economics of the proposed mmW-MBN framework}
Fig. \ref{ffig8} provides a design guideline to manage pricing and the cost of the cooperative mmW-MBN for the MNOs. In this figure, the cost of cooperation per MNO is shown as the price per sub-channel $q$ and the weighting parameter $\kappa$ vary, for $N=3$ MNOs and $M=15$ SBSs. The weighting parameter $\kappa_m$, which is first defined in \eqref{opt1:a} allows each MNO to control the cost of its backhaul network, with respect to $q$ that is determined by other MNOs. \emph{Here, we explicitly define the backhaul cost for an MNO $n \in \mathcal{N}$, as the total money that MNO $n$ must pay to other MNOs for receiving backhaul support to the SBSs in $\mathcal{M}_n$}. A larger $\kappa_m$ implies that the MNO has less incentive to cooperate with other MNOs. Hence, as shown in Fig. \ref{ffig8}, no cooperation will happen between MNOs, as both $q$ and $\kappa_m$ increase, labeled as the \emph{no cooperation} region. As an example, if the budget of an arbitrary MNO $n$ is $\$500$ and $q=\$10$ is chosen by other MNOs, from Fig. \ref{ffig8}, we can see that MNO $n$ must choose $\kappa_m \geq 40\, \text{Mbps}/\$$ in order to keep the cost less than its budget. In addition, Fig. \ref{ffig8} will provides a systematic approach to determine a suitable pricing mechanism for an MNO, if the model parameter $\kappa_m$ and the budget for other MNOs are known. This initial result can be considered as a primary step towards more complex models, which may consider dynamic pricing policies and competing strategies for MNOs.}
\begin{figure}[!t]
	\centering
	\centerline{\includegraphics[width=\columnwidth]{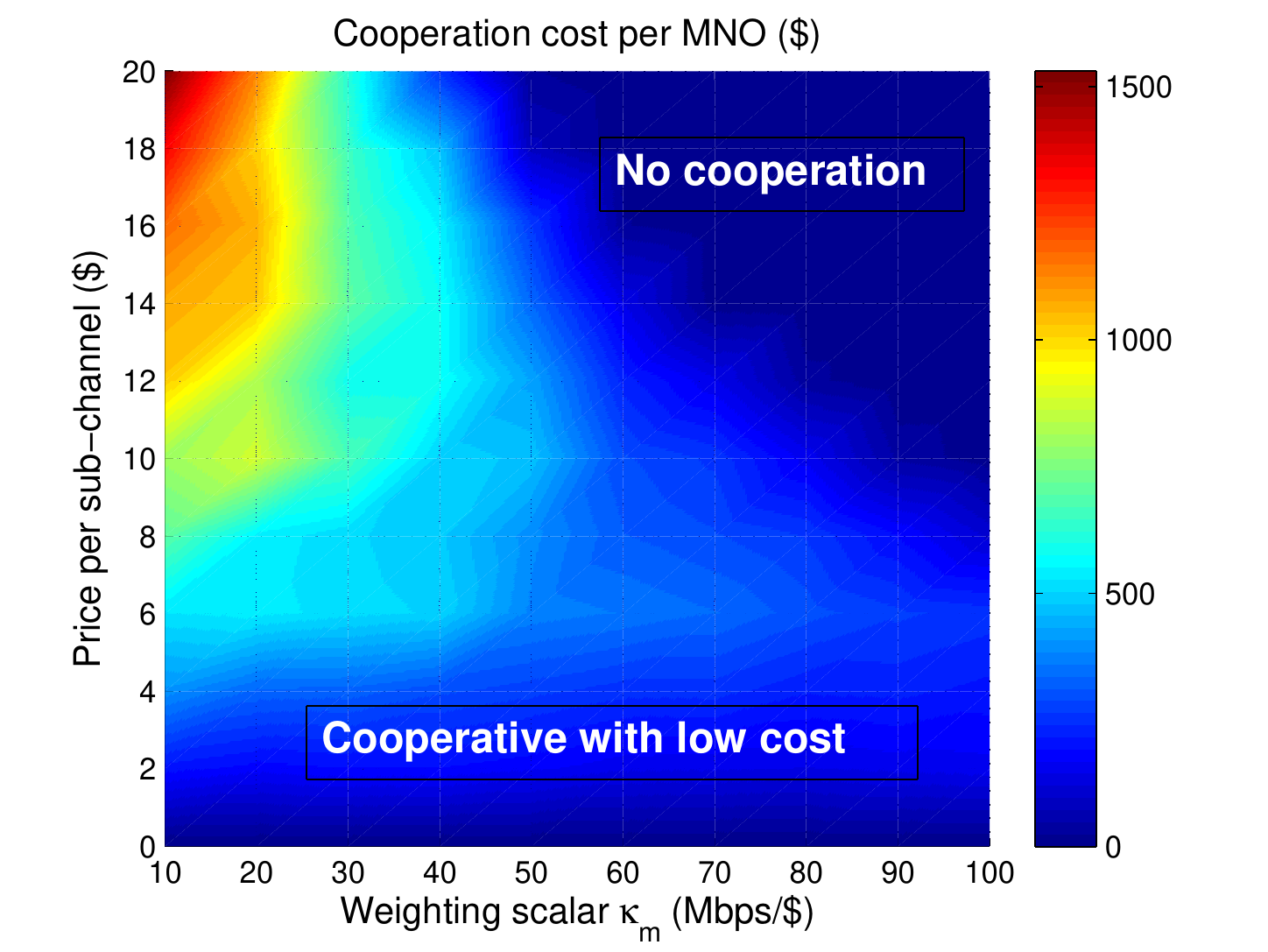}}\vspace{-.2cm}
	\caption{\small Cooperation cost per MNO as a function of both sub-channel price and the weighting parameter $\kappa_m$.}\vspace{0em}
	\label{ffig8}
\end{figure}
\textcolor{black}{
Moreover, we note that the economic gains of the proposed cooperative framework are indirectly reflected in the performance gains of the proposed scheme, compared to the non-cooperative and random allocation approaches, as shown in Figs. \ref{ffig3}-\ref{ffig6}. Such an increase in the data rate of the backhaul network, resulting from the cooperative framework, will provide additional revenues for the MNOs, either by offering services with higher QoS to the users, or by increasing the users served by each SBS. Here, we note that the revenue for each MNO explicitly depends on the cost of maintenance per SBS, leasing the spectrum, deployment of SBSs, providing power supply for SBSs, service plans by MNOs and other specific metrics that may differ from one geographical area to another. Therefore, there is no direct and general mechanism to map the physical layer metrics, such as rate into revenue. However, such a mapping is definitely being used by the economic departments of global operators to define their KPI performance metrics.}

\textcolor{black}{Consequently, in Fig. \ref{ffig8}, we have shown the robustness of the proposed framework with regard to the pricing mechanisms. In fact, we have shown that the proposed resource management framework allows MNOs to choose whether to cooperate or not, depending on the system metrics, including the rate, their available budget, and the pricing policy by other~MNOs.}
\vspace{-0em} 
\textcolor{black}{\subsection{Snapshot of the mmW-MBN}}

\begin{figure*}[t!]
	\centering
	\begin{subfigure}[b]{7cm}
		\includegraphics[width=\textwidth]{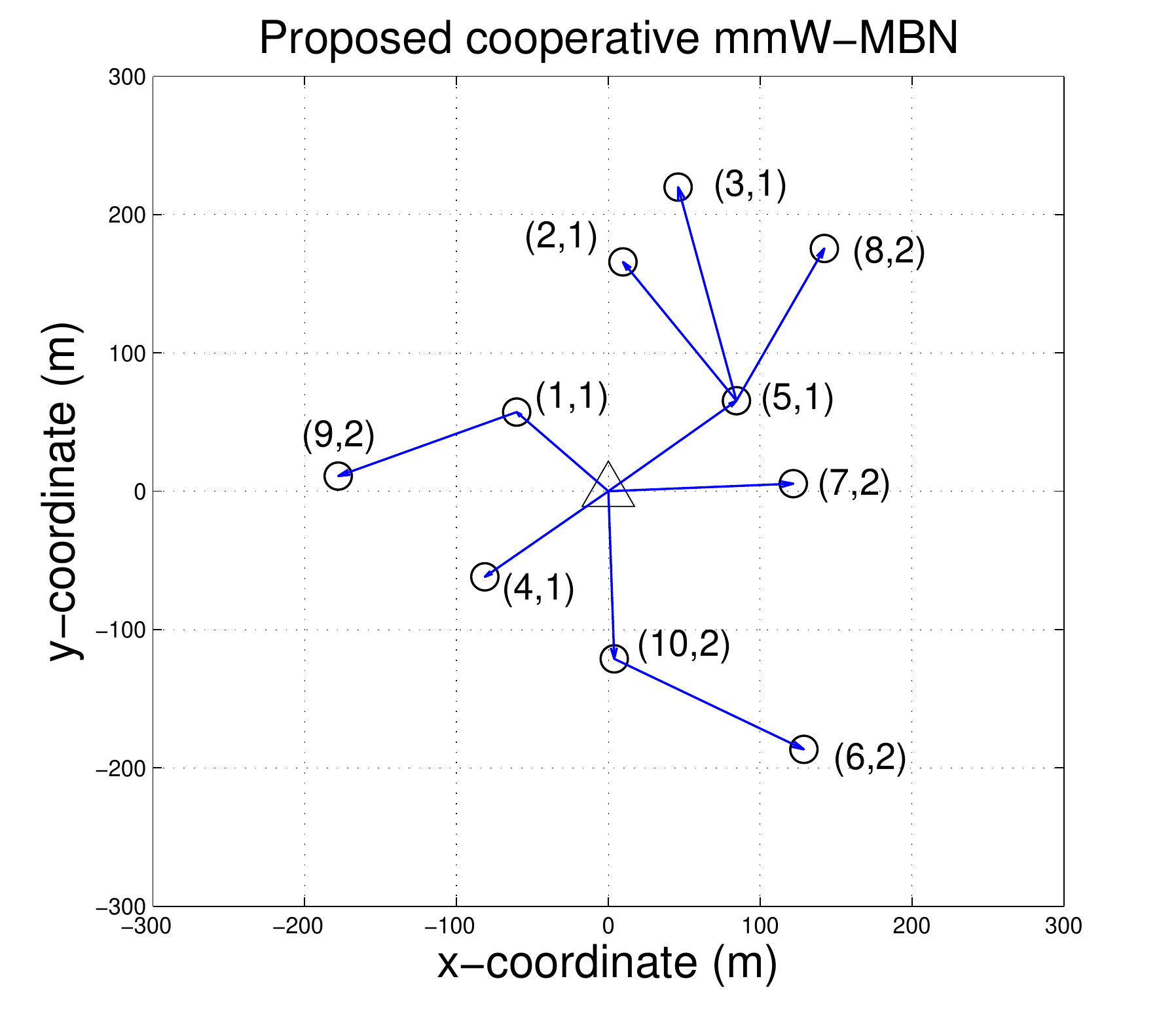}\vspace{-0.2cm}
		\caption{Proposed cooperative mmW-MBN}
		\label{snapshot1}
	\end{subfigure}
	~\hspace{-.4cm} 
	\begin{subfigure}[b]{7cm}
		\includegraphics[width=\textwidth]{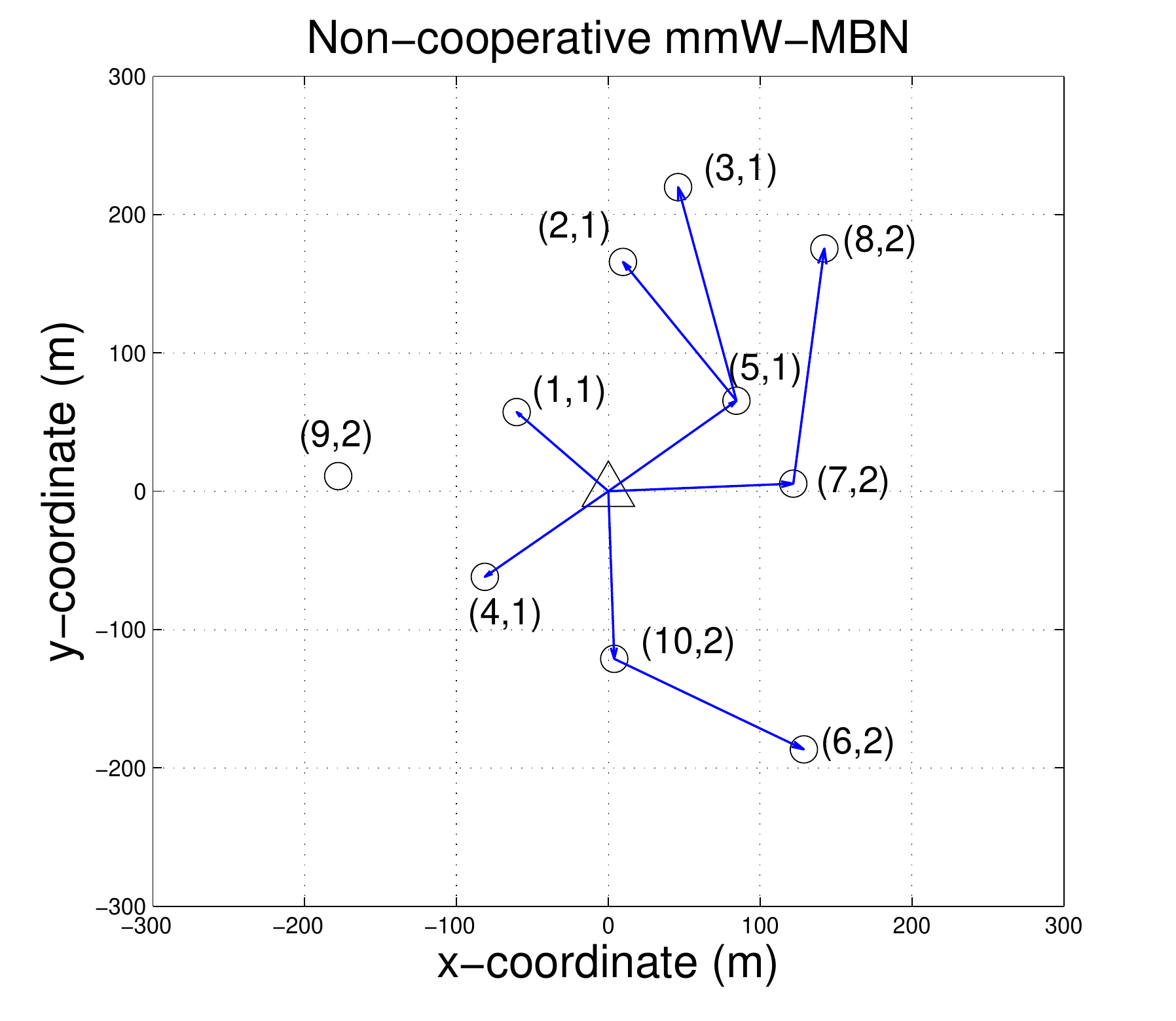}\vspace{-0.2cm}
		\caption{Non-cooperative mmW-MBN}
		\label{snapshot2}
	\end{subfigure}
	\vspace{-.1cm} 
	\caption{\small A snapshot of multi-hop mmW backhaul network via the proposed cooperative scheme and the non-cooperative baseline approach.}\label{snapshot}
\end{figure*}

Fig. \ref{snapshot} shows a snapshot of the mmW multi-hop backhaul network (mmW-MBN) for both the proposed cooperative scheme and the non-cooperative baseline approach. In this figure, each circle shows a small cell base station (SBS) and corresponding pair $(m,n)$ means SBS $m$ belongs to MNO $n$ ($m \in \mathcal{M}_n$). Moreover, the MBS is shown by a triangle. For illustration purposes, we show the network for $N=2$ MNOs and a total of $M=10$ SBSs, and the quota for each A-BS is $Q_m=5$.

From Fig. \ref{snapshot2}, we first observe that  SBS $9 \in \mathcal{M}_2$ is not connected to the non-cooperative mmW-MBN, since no other SBS belonging to MNO $2$ is located within SBS $9$'s communication range. That is, MNO $2$ must increase the density of its SBSs to provide ubiquitous backhaul connectivity. However, deploying additional SBSs will increase the costs for the MNO, including site rental costs, power consumption, and cell maintenance, among others. In contrast, the proposed cooperative scheme provides backhaul support for the SBS $9 \in \mathcal{M}_2$ via SBS $1 \in \mathcal{M}_1$ that belongs to the MNO $1$, as shown in Fig. \ref{snapshot1}. Second, Fig. \ref{snapshot} shows that SBS $8 \in \mathcal{M}_2$ is connected to A-BSs $7 \in \mathcal{M}_2$ and $5 \in \mathcal{M}_1$, respectively, in the non-cooperative and proposed cooperative mmW-MBNs. We can easily observe that the proposed cooperative scheme provides a shorter path via a two-hop backhaul link for the SBS $m=8$, compared to the non-cooperative approach.
\vspace{-0em}
\subsection{Complexity analysis}
\begin{figure}[!t]
	\centering
	\centerline{\includegraphics[width=10cm]{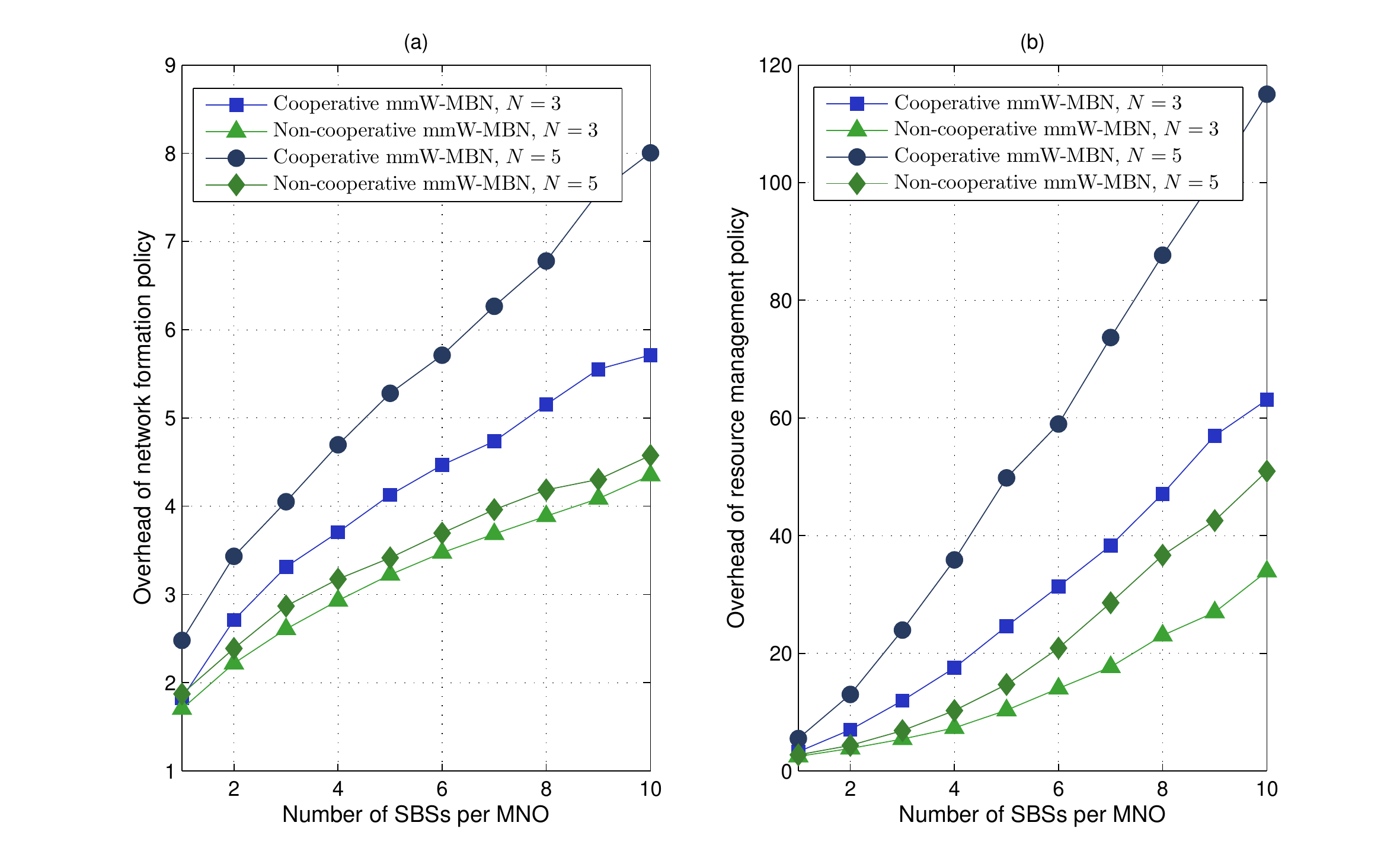}}\vspace{-.5cm}
	\caption{\small The average overhead of the network formation and resource allocation algorithms.}\vspace{-1em}
	\label{sim2}
\end{figure}\vspace{0em}
In Fig. \ref{sim2}, the average signaling overhead of the proposed network formation and the resource allocation algorithms are analyzed, respectively, in (a) and (b). Here, the overhead captures the number of messages that must be exchanged between A-BSs and D-BSs. Fig. \ref{sim2}.a shows that the overhead of the proposed network formation policy increases with the number of SBS per MNOs, since the sets of A-BSs and D-BSs grow as more SBSs are deployed. However, we can see that the algorithm converges fast for all network sizes. Moreover, it can be observed that the proposed cooperative approach increases the overhead by $28 \%$ for $M=18$ SBSs. This is because the proposed approach allows each SBS to communicate with more number of SBSs, compared to the non-cooperative scheme. Similarly, in Fig. \ref{sim2}.b, the overhead of the resource allocation algorithm increases as the number of SBSs increases. 
Regarding the complexity of the proposed approach, we note the followings:1) an optimal solution will require an exponential complexity which is not tractable for dense mmW network deployments, while the proposed approach yields a close-to-optimal performance, as shown in Fig. \ref{ffig4}, while requiring a manageable signaling overhead, 2) in this work, we have considered MBS as the only gateway, with a fiber backhaul, to the core network. Therefore, the scenario that is considered in our simulation is an extreme case, as in practice, there will be more than one gateway for up to $M=65$ SBSs. Clearly, increasing the number of gateways will reduce the size of the problem, i.e., the number of SBSs to be managed, and 3) communication signals required by the proposed scheme will be incorporated within the common control signals of the system.
\section{Conclusion}\label{conclusion}
In this paper, we have proposed a novel distributed backhaul management approach for analyzing the problem of resource management in multi-hop mmW backhaul networks. In particular, we have formulated the problem within a matching-theoretic framework composed of two, dependent matching games: a network formation game and a resource management game. For the network formation game, we have proposed a deferred acceptance-based algorithm that can yield a two-sided stable, Pareto optimal matching between the A-BSs and D-BSs. This matching represents the formation of the multi-hop backhaul links. Once the network formation game is determined, we have proposed a novel algorithm for resource management that allocates the sub-channels of each A-BSs to its associated D-BSs. We have shown that the proposed resource management algorithm is guaranteed to converge to a two-sided stable and Pareto optimal matching between the sub-channels and the D-BSs. Simulation results have shown that the proposed cooperative backhaul framework provides substantial performance gains for the network operators and incentivizes sharing of the backhaul links.

\bibliographystyle{IEEEtran}
\bibliography{references}

\begin{IEEEbiography}[{\includegraphics[width=1in,height=1.25in,clip,keepaspectratio]{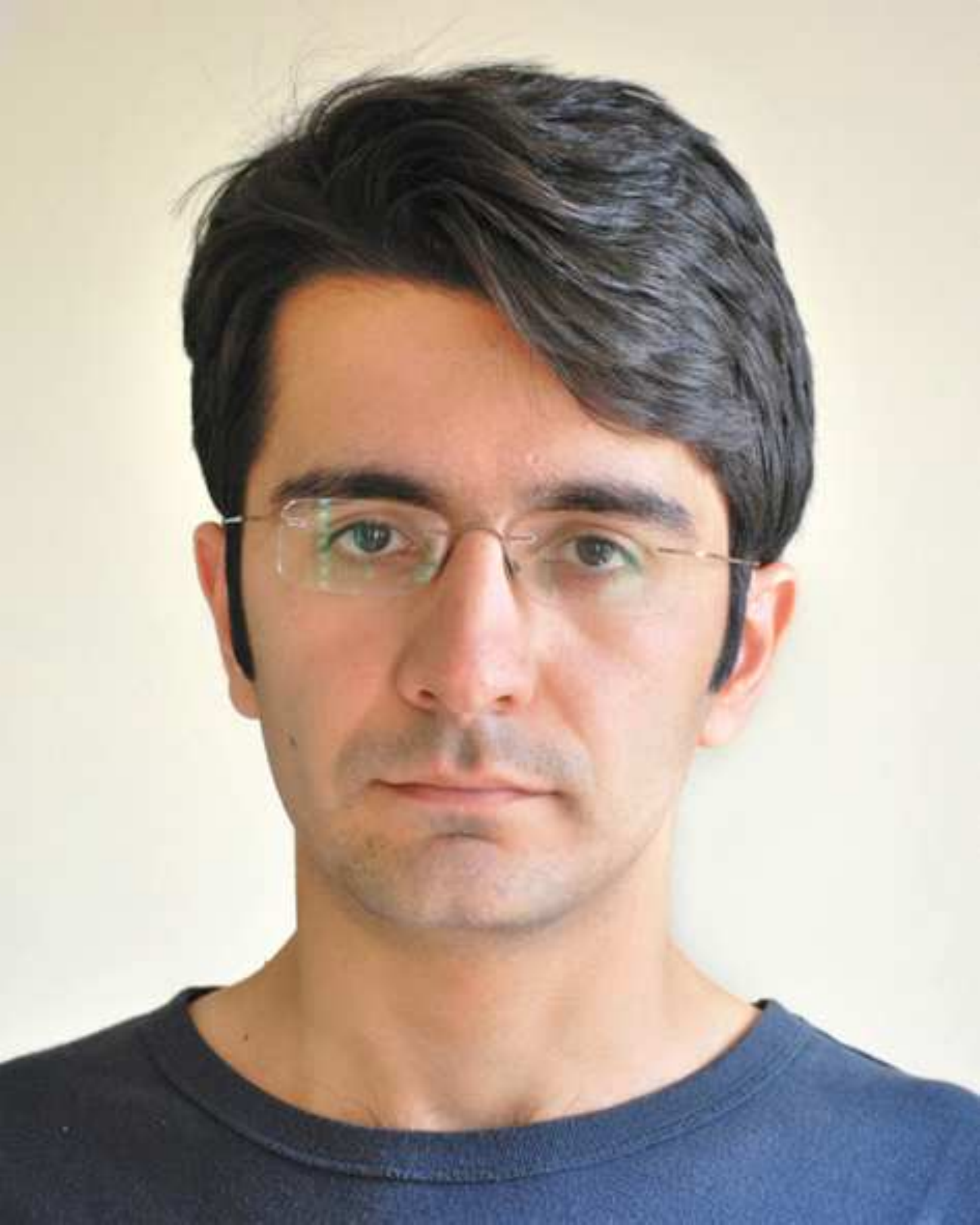}}]{Omid Semiari}(S'14)  received the B.Sc. and M.Sc. degrees in communication systems from University of Tehran in 2010 and 2012, respectively. He is currently a PhD candidate at the Bradly department of Electrical and Computer Engineering at Virginia Tech. In 2014, he has worked as an intern
	at Bell Labs, on anticipatory, context-aware resource management. In 2016, he has joined Qualcomm CDMA Technologies (QCT) for a summer internship, working on LTE-Advanced modem design. Mr. Semiari is the recipient of
	several research fellowship awards, including DAAD (German Academic Exchange Service) scholarship and NSF
	student travel grant. He has actively served as a reviewer for 
	flagship IEEE Transactions and conferences and participated as the technical program committee (TPC) member for a variety of workshops at IEEE conferences, such as ICC and GLOBECOM. His research interests
	include wireless communications and networking, millimeter wave communications, context-aware
	resource allocation, matching theory, and signal processing.\vspace{-1cm}
\end{IEEEbiography}

\begin{IEEEbiography}[{\includegraphics[width=1in,height=1.25in,clip,keepaspectratio]{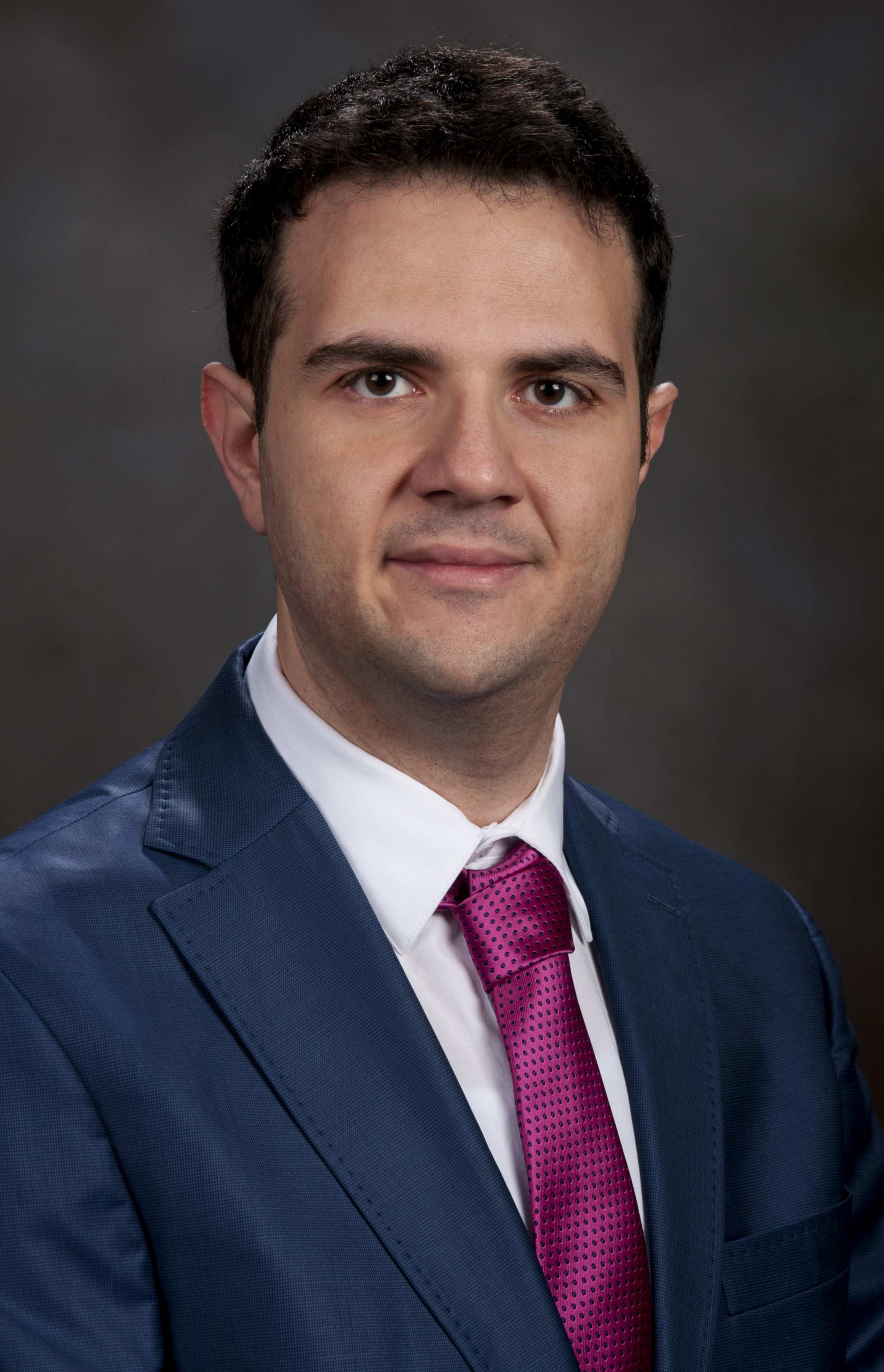}}]{Walid Saad}(S'07, M'10, SM'15) received his Ph.D degree from the University of Oslo in 2010. Currently, he is Associate Professor at the Department of Electrical and Computer Engineering at Virginia Tech, where he leads the Network Science, Wireless, and Security (NetSciWiS) laboratory, within the Wireless@VT research group. His research interests include wireless networks, game theory, cybersecurity, unmanned aerial vehicles, and cyber-physical systems. Dr. Saad is the recipient of the NSF CAREER award in 2013, the AFOSR summer faculty fellowship in 2014, and the Young Investigator Award from the Office of Naval Research (ONR) in 2015. He was the author/co-author of five conference best paper awards at WiOpt in 2009, ICIMP in 2010, IEEE WCNC in 2012, IEEE PIMRC in 2015, and IEEE SmartGridComm in 2015. He is the recipient of the 2015 Fred W. Ellersick Prize from the IEEE Communications Society. In 2017, Dr. Saad was named College of Engineering Faculty Fellow at Virginia Tech. From 2015 – 2017, Dr. Saad was named the Steven O. Lane Junior Faculty Fellow at Virginia Tech. He currently serves as an editor for the IEEE Transactions on Wireless Communications, IEEE Transactions on Communications, and IEEE Transactions on Information Forensics and Security.
\end{IEEEbiography}

\begin{IEEEbiography}[{\includegraphics[width=1in,height=1.25in,clip,keepaspectratio]{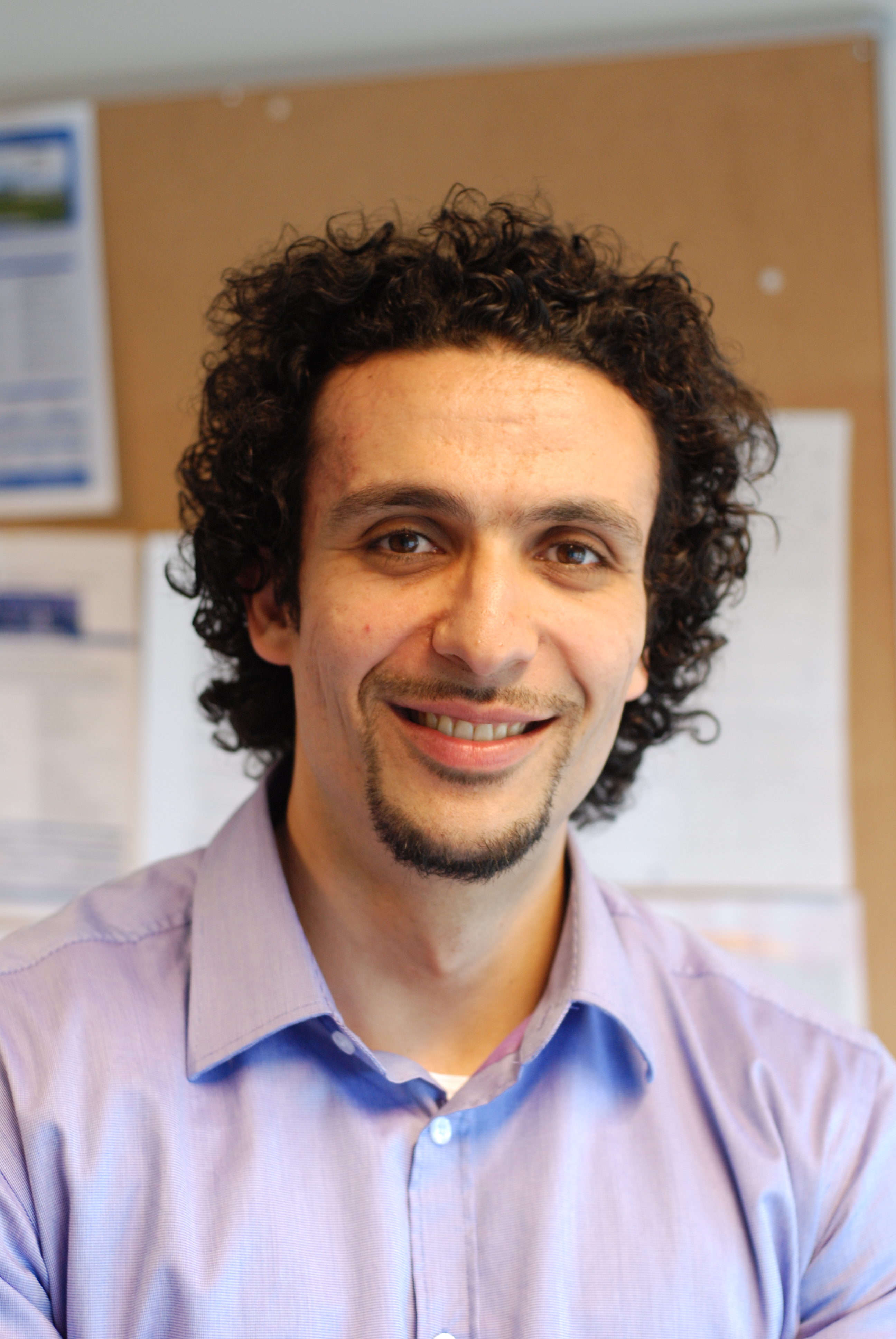}}]{Mehdi Bennis} (Senior Member, IEEE) received his M.Sc. degree in 
	Electrical Engineering jointly from the EPFL, Switzerland and the 
	Eurecom Institute, France in 2002. From 2002 to 2004, he worked as a 
	research engineer at IMRA-EUROPE investigating adaptive equalization 
	algorithms for mobile digital
	TV. In 2004, he joined the Centre for Wireless Communications (CWC) at 
	the University of Oulu, Finland as a research scientist. In 2008, he was 
	a visiting researcher at the Alcatel-Lucent chair on flexible radio, 
	SUPELEC. He obtained his Ph.D. in December 2009 on spectrum sharing for 
	future mobile cellular systems. Currently Dr. Bennis is an Adjunct 
	Professor at the University of Oulu and Academy of Finland research 
	fellow. His main research interests are in radio resource management, 
	heterogeneous networks, game theory and machine learning in 5G networks 
	and beyond. He has co-authored one book and published more than 100 
	research papers in international conferences, journals and book 
	chapters. He was the recipient of the prestigious 2015 Fred W. Ellersick 
	Prize from the IEEE Communications Society, the 2016 Best Tutorial 
	Prize from the IEEE Communications Society and the 2017 EURASIP Best paper Award for the Journal of wireless communications and networks..
	Dr. Bennis serves as an editor for the IEEE Transactions on Wireless Communication
\end{IEEEbiography}

\end{document}